\documentclass[preprint,11pt]{amsart}

\newcommand{\secref}[1]{Section~\ref{#1}}

\makeatletter
\newcommand\footref[1]{\protected@xdef\@thefnmark{\ref{#1}}\@footnotemark}
\makeatother

\newcommand{\mscr}{\mathscr}
\newcommand{\mb}{\mathbf}
\newcommand{\mc}{\mathcal}
\newcommand{\mf}{\mathfrak}
\newcommand{\mbb}{\mathbb}
\newcommand{\bs}{\begin{split}}
\newcommand{\es}{\end{split}}
\newcommand{\eqr}{\eqref}

\newcommand{\tta}{\theta}

\newcommand{\eps}{\epsilon}

\newcommand{\bal}{\begin{align}}
\newcommand{\beq}{\begin{equation}\begin{aligned}}
\newcommand{\eeq}{\end{aligned}\end{equation}}
\newcommand{\closure}[1]{\mkern 1.5mu\overline{\mkern-.2mu#1\mkern-1.8mu}\mkern 1.5mu}

\let\norm=\enVert

\usepackage{amsfonts}
\usepackage{amsmath}
\usepackage{amssymb,latexsym}
\usepackage{mathabx}
\usepackage{dsfont}
\usepackage{mathtools}
\usepackage{accents}
\usepackage{array}
\usepackage{float}
\usepackage[toc,titletoc]{appendix}
\usepackage{longtable}
\usepackage{graphicx}
\usepackage{caption}
\usepackage{subcaption}
\usepackage{multirow}
\usepackage[T1]{fontenc}
\usepackage[utf8]{inputenc}
\usepackage[font=singlespacing]{caption}
\usepackage{tabularx}
\usepackage{enumerate}
\usepackage[normalem]{ulem}
\usepackage{hyperref}
\hypersetup{colorlinks=true,linkcolor=blue,citecolor=red,urlcolor=blue}
\usepackage{mathrsfs}
\usepackage{young}
\usepackage{ytableau}
\usepackage[vcentermath]{youngtab}
\usepackage{braket}
\usepackage{cancel}
\usepackage{tensor}
\usepackage{units}
\usepackage{rotating}
\usepackage{fixmath}
\usepackage[top=1in, bottom=1in, left=1.25in, right=1.25in]{geometry}
\usepackage{setspace}

\usepackage{pgfplots}
\usepgfplotslibrary{polar}
\usepgflibrary{shapes.geometric}
\usetikzlibrary{calc}

\pgfplotsset{my style/.append style={axis x line=middle, axis y line=
           middle, xlabel={$x$}, ylabel={$f(x)$}, axis equal }}
\newtheorem{thm}{Theorem}[section]
\newtheorem{lem}[thm]{Lemma}
\newtheorem{prop}[thm]{Proposition}
\newtheorem{cor}[thm]{Corollary}
\newtheorem{defin}[thm]{Definition}
\newtheorem{rmk}{Remark}

\newenvironment{example}[1][Example]{\begin{trivlist}
\item[\hskip \labelsep {\bfseries #1}]}{\end{trivlist}}

\DeclareMathOperator{\sgn}{sgn}
\begin{document}

\title{Semicontinuous Banach Spaces for Schr\"{o}dinger's Eq. with Dirac-$\delta'$ Potential}

\author{B. Button}
\address{University of Houston-Victoria\\ School of Arts \& Sciences \\ 3007 N. Ben Wilson St.\\ Victoria, TX-77901}
\email{buttonb@uhv.edu}
\date{\today}
\maketitle

\begin{abstract}
Schr\'{o}dinger's equation with distributional $\delta$, or $\delta'$ potentials has been well studied in the past. There are challenges in simultaneously addressing some of the inherent issues of the system: The functional operator cannot exist entirely within the standard $L^2$ Hilbert spaces. On differentiable manifolds, the domain of the free kinetic energy operator is in the space of harmonic forms. Locally, by the Hodge decomposition theorem and the standard distributional calculus, the space of functionals of a $\delta$ or $\delta'$ potential must be orthogonal to the free kinetic energy operator. Restricting to semicontinuous topologies presents opportunities to address these, and other issues. We develop, in great detail, a formalism of Banach spaces with semicontinuous topologies, and their properties are extensively defined and studied. For $C(\closure{\mbb{R}})$ functions, the spaces are indistinguishable. The semicontinuous analogs of the $L^P$ spaces, are nontrivial and result in a dense topologically continuous embedding of the semicontinuous $L^p$ spaces into the semicontinuous $C(\closure{\mbb{R}})$ spaces. Here, certain classes of distributions may be inverted in terms of their primitive functions. Also many operators are inherently self adjoint. We define equivalence relations between the cohomology classes of distributions and derivatives of their associated primitives on local sections of $\closure{\mbb{R}}$. Here Hamilton's equations are canonical, and define a connection on the fibers of the base space. Semicontinuity provides a resolution to the above domain and interaction problems, and easily integrable Feynman functional. We arrive at a compatible domain which is Krein ($\mf{H}$) over disjoint components of $\closure{\mbb{R}}$. The subspaces of $\mf{H}$ are isomorphic to the semicontinuous Hilbert spaces of the Hamiltonian.  
\end{abstract}



\section{Introduction}\label{intro}
The study of quantum mechanics necessitates the study of the self adjoint Hamiltonian operator on some Hilbert space, $\mc{H}=L^2(\mbb{R}^n)$ for example. In one dimension, for some wave function $\psi\in\mc{H}$, the Hamiltonian (energy operator) acting on $\psi$ is Sch\"{o}dinger's differential equation given by,
\beq\begin{aligned}\label{qmSchrdeq}
-i\frac{d}{dt}\psi=\hat{H}\psi=E\psi,
\end{aligned}\eeq
where $E$ is the eigenvalue of the energy operator, $\psi=\psi(x,t)$. Since this is usually defined on $L^2$, the standard physics inner product notation for this coupled integral-differential equation is the Dirac bra-ket, $\bra{\psi^\ast}-\frac{d}{dt}\ket{\psi}=\bra{\psi^\ast}\hat{H}\ket{\psi}$. If the energy of the system is constant in time, then $\hat{H}$ is given as
\beq\begin{aligned}\label{qmHam}
\hat{H}&=\hat{P}^2+V(\hat{x})\\
\end{aligned}\eeq 
where the $\hat{P}$ is the Hermitian momentum operator is $\hat{P}=-i\frac{\hbar}{\sqrt{2m}}\frac{\partial}{\partial x}$. Here $m$ is the mass of a point-like particle and $V(\hat{x})$ is a time independent potential. We will denote the free Hamiltonian operator $\left(V(\hat{x})=0\right)$ by $\hat{H}_f$.

Over the past couple of decades, there has been a considerable amount of work done by both the physics and mathematics community for cases where the potential is highly singular, and in particular point supported, such as the Dirac-$\delta$ or derivatives of the Dirac-$\delta$. Thus one is left to make sense of a Hamiltonian operator of the sort
\begin{align}\label{hamdelta}
 \hat{H}=-\frac{\hbar^2}{2m}\frac{\partial^2}{\partial x^2}+\delta(x)\\
 \shortintertext{or,}
\label{hamdelta2}
 \hat{H}=-\frac{\hbar^2}{2m}\frac{\partial^2}{\partial x^2}+\delta'(x).
 \end{align}  
We will work explicitly on the extended real line $\closure{\mbb{R}}^1$, we use the notation for partial derivatives $\frac{\partial}{\partial x}$ to denote the differential operator. This will facilitate later discussions when we discuss \eqr{hamdelta2} in terms of differential forms on closed manifolds. Furthermore, not all results will be here, generalize to higher dimensions. Though it should be readily apparent which results admit higher dimensional generalizations.

There are several considerations which motivates our study of semicontinuous spaces. The collective set of motivations aim for both mathematical rigor (as much as possible) and relevance to applications in physics. As such, the work here attempts to bridge the difference between mathematical theory and theoretical physics. The results which follow are not quite those of the usual $L^p(\mbb{R})$ theory, but they are not so different as to be completely unrecognizable from it. The formalism developed herein has some unexpected, but pleasant properties. The properties are themselves noteworthy in their own right, but also have potential for use in low dimensional condensed matter systems (in particular graphene sheets) and could possibly produce non-trivial results in \textit{AdS/CFT} or string theory itself. We will make some general comments in the summary regarding applications to interactions in perturbative string theory, as well as other avenues for future investigations. As such, in~\secref{deltaprm}, we choose to apply the semicontinuous spaces to analyze the basic quantum system defined by Eq.~\eqr{hamdelta2} in terms of Feynman's functional integral. We reserve applications to specific to other systems, \textit{i.e.} string theory and holography for future work. It is with these considerations in mind, as well as the notable differences of $\delta$ and $\delta^\prime$ potentials in 2 and 3 dimensions \cite{albeverio2012solvable} and specifically, the use of spherical/polar coordinates, that we do not to generalize beyond the discussion beyond one dimension. However, provided that our postulates hold, there is no \textit{a priori} reason we expect that higher dimensional generalizations will necessarily fail. The issue (as is the case with all regularization methods), is whether or not the regularized system is representative of the initial system. 

From a pure mathematics perspective, we address two particularly troublesome issues regarding the system defined by Eq.~\eqr{hamdelta2}, from which, we may define a method to enable one to more completely utilize the Feynman path integral in similar cases. We construct a formalism which is \textit{sufficient} to accommodate for the problems that,
\begin{enumerate}
\item Although the space of test functions for $\delta$ is dense in the space of $L^p(\mbb{R})$ functions, the space of test functions for $\delta$ (and derivatives of $\delta$) is not equivalent to the entirety of any $L^p$ space, for any $1\leq p<\infty$.
\item  In general, the kinetic energy operator and the "potentials" of Eq.~\eqr{hamdelta} or Eq.~\eqr{hamdelta2} do not act as maps from the same base space to the same target space.
\end{enumerate}
To elaborate on (1), methods of approximations or limiting sequences which approach $\delta$ or $\delta^\prime$ potentials have meaning in $L^2$ \textbf{except} at the limit point itself, where the $L^2$ functions actually become true singular distributions. Closure in functional spaces (and thus self adjoint-ness) is lost. Extensions to $L^2$ or even Sobolev subspaces are not sufficient in such circumstances\cite[Thm. 8.27]{reedsimon1}. 

Further elaborations on (2) have two complementary facets. As a distribution (i.e. at the limit point of some approximation scheme for $\delta,\delta^\prime$), the "potential" becomes a map from the space of Schwartz functions ($\mc{S}$) to the real numbers, whereas the differential operator is a linear transformation from one functional vector space to another. The complementary issue arises if one attempts to place the singular Hamiltonians on a differentiable manifold. Here, one may view the systems above as linear functionals on the space of differential forms. Local arguments for singular operators on differential forms are still subject to the Hodge decomposition theorem. As a result, the second order differential operator and the singular potential must belong to orthogonal spaces. The vector on which the kinetic energy operator acts, is necessarily orthogonal to that of the potential, unless they are the same $p$-form degree. This prohibits the system from being a truly interacting system. We will discuss each of these in more detail shortly. 

Also, each component of our formalism is $necessary$ in the sense that the formalism is self consistent, while simultaneously addressing both (1) and (2), as well as some other points which we will encounter along the way. The space of singular distributions do not have the same notion of "domain" as linear operators or linear transformations. In an abuse of language, we will often reference the collective domain of the linear functional, \eqr{hamdelta} or \eqr{hamdelta2}.

The main results of the paper span the sets of different tools used to address each of the above points. Topological measure spaces are constructed in such a way as to address the particular domain incompatibility between the distribution and the differential operator components of Eq.~\eqr{hamdelta2}. We define spaces of semi-continuous function(al)s which cannot distinguish between $C(\closure{\mbb{R}})$ and $L^p(\closure{\mbb{R}})$ functions. Therefore we have no need to extend our results to subsets of $L^p$ spaces for self adjointness. A positive consequence of our construction is that many operators are inherently self adjoint. The mapping to semicontinuous spaces produces subspaces of semicontinuous $L^p$ functions which are orthogonal in regards to left versus right semi-continuity. The defined semi-continuous measure spaces allows the freedom to define an equivalence class between $\delta$ and derivatives of functions which may be considered as a primitive function of $\delta$, such as the Heaviside function $\theta$. This in turn affords one the ability to define a mapping of $\delta'$, via equivalence class identifications, to the cohomology class of harmonic forms. Here, the system becomes a genuinely interacting system. It is shown that the equivalence class mapping is locally canonical. Our example discussed in ~\secref{deltaprm}, the formalism is applied to Feynman's path integral. The net result of the formalism, collectively broadens the applicability of the Feynman path integral to include exponentiation of full Hamiltonian. The resulting function space is Krein. It is known that Krein spaces have subspaces which are isomorphic to $L^p$ spaces of functions. 

\subsection{Notation and Conventions}\label{conventions}
Here we pause in order to state the conventions and notation used throughout the paper. In later sections we will discuss the Hamiltonian as a functional on a differentiable (psuedo-Riemannian) manifold, where differential geometric structures are relevant, and it will be necessary to make distinctions between various differential operators. Thus, given a differentiable function $f$ then, $\frac{\partial}{\partial x}f:=dx\frac{\partial f}{\partial x}$ is a covector (1-form) in some cotangent space at the point $x$, whose component is $\frac{\partial f}{\partial x}$. Exterior differentiation and codifferentiation will exclusively be denoted by $\mathbold{d}f$ and $\mathbold{\delta}f$, respectively. Then, $\frac{\partial^2}{\partial x^2}$ will be implicitly defined by the Laplace-Beltrami operator; $\mathbold{\Delta}:=\mathbold{\delta d}+\mathbold{d\delta}$. We will adopt the general notation of $D$ when regarding differentiable set definitions or where we wish regard differentiation more colloquially, and will assume the applicable derivative to be implicit. For functions (say $f$) used in equations, we generally denote derivatives with respect to their arguments as $f'(x)$, and $f'$ for distributional/functional derivatives, and commonly use either interchangeably when there is no danger of ambiguity. 

The space of Schwarz functions, the space of smooth functions such that $f(x)$ and all its derivatives go zero faster than $x^n$ for all $n\in\mbb{Z}$, is denoted by $\mc{S}$. The topological dual of $\mc{S}$, the space of tempered distributions, is denoted by $\mc{S}^\prime$. We shall also commonly, but not exclusively, denote distributions in a manner similar to: $f_{\delta_x}=f(0)$, for some $f\in\mc{S}$, and the Dirac-$\delta$ with point support at $x$. 

A Borel measurable space over some universal set $(\Sigma, \mathscr{B})=\mathscr{B}(\Sigma)$ with total variation measure, $\mu$ (or sometimes $\nu$), defines a measure space $(\Sigma, \mathscr{B},\mu)=(\mathscr{B}(\Sigma),\mu)$. We denote the space of bounded linear functions over the previously given measure space by $\mc{L}(\mathscr{B},\mu)$. We almost exclusively have $\Sigma=\closure{\mbb{R}}$, and so in this case we will omit the universal set from the notation. We denote the left (respectively right) semicontinuous spaces of bounded linear functions by with left-semicontinuous measure $\mu_L$ (resp. right-semicontinuous measure $\mu_R$) by $\mc{L}_L(\mathscr{B},\mu_L)$ (resp. $\mc{L}_R(\mathscr{B},\mu_R$). Left semicontinuity (resp. right semicontinuity) is defined to be the continuous one-sided measure approaching \textit{from the left} (resp. right). For example, if $[a,b]\in\closure{\mbb{R}}$, then the left-semicontinuous Borel measure will be given by $\mu_L:=\mu((a,b))\cup\mu\left\{b\right\}=\mu\left((a,b]\right)$, where the half-open interval notation is as expected. These are just Stieltjes measures over half-open Borel sets. 

The standard Lebesgue measure denoted as $\lambda$. Measures of functions under Lebesgue equivalence class identifications are generally understood with respect to the usual Lebesgue-Stieltjes measure. Left (resp. right) continuous Lebesgue (Lebesgue-Stieltjes (L-S)) measures are denoted by $\lambda_L$ (resp. $\lambda_R$).  Generally, whether we have measures $\mu_{L,R}$ or $\lambda_{L,R}$ (where "$L,R$" refers collectively to left or right semicontinuous measures), sets of measure zero under Lebesgue measure, can have non-zero measure under L-S measures. We will be more precise about the semicontinuous measures and measure spaces in \secref{semi} below. 

\subsection{Further Comments on the Hamiltonian Functional}\label{commentsonHam}
If Eq.~\eqr{hamdelta} is to act on a wave function $\psi$, as in Eq.~\eqr{qmSchrdeq}, the resulting expression is ill-defined as a differential equation. In particular $\frac{\partial}{\partial x}$ is a differential operator, which is a linear transformation $D:\mc{L}\to\mc{L}'$, where $\mc{L}$ denotes an arbitrary differentiable manifold, vector or linear functional space with domain $\mc{D}_{\mc{L}}$. In particular, the kinetic energy operator is the mapping $D^2=(D\circ D):\mc{L}\to\mc{L}'\to\mc{L}^\second$. By implicitly assuming $D^2\sim \mathbold{\Delta}$, then $D^2$ will map $p$-forms in $\mc{L}\to p$-forms in $\mc{L}^\second$ such that, for local neighborhoods $U,V$ of some base spaces $M,N$ respectively, we require $\mc{L}(M)\large|_U\cong\mc{L}^\second\large(N)\large|_V$.  However $\delta_x$ only has rigorous meaning either as a tempered distribution ($\delta\in\mc{S}'$) or as a measure ($\mu_\delta\in \mc{B}\left(\Sigma\right)$), with $\Sigma$ some measurable $\sigma$-algebra manifold. But what sort of object (or linear space) is Eq.~\eqr{hamdelta2} acting on? A rigid interpretation of Eq.~\eqr{hamdelta2}, implies that $\frac{\partial}{\partial x}\delta(x)=dx\frac{\partial \delta(x)}{\partial x}$, which makes the "potential" a 1-form. But then we could have no scalar wave function solutions which are mapped to a common space (although direct product/sum spaces can be constructed). By the Hodge decomposition theorem, the solution space of harmonic $0$-forms is orthogonal to the solution space of the codifferential on 1-forms. In this case, the interaction between the free kinetic energy operator and the potential become independent, and thus completely decouple by orthogonality.

This is particularly troublesome with respect to Feynman path integrals, where one would like to have solutions to the functional integral of the form 
\beq\label{arbpathint}
\int{\mc{D}q^i \mc{D}p^i}~\psi^\ast (q,p,t)\exp^{i\int{dt}\mscr{H}(q,q_i,p_i;t)}\psi(q_i,p_i;t_i)=\delta(q-q_i)\delta(t-t_i),
\eeq  
where $\mscr{H}(q,q_i,p_i;t)$ is the Hamiltonian density functional, and $\psi(q)$ are required to satisfy (to at least first order in $t$) $\psi^\ast (q,t) \hat{H}\ket{\psi(q_i;t_i)}=\delta(q-q_i)\delta(t-t_i)$. Classically, this is interpreted as $\braket{\psi^\ast|\hat{H}|\psi}=0$. 

Assuming that $\psi$ is a 0-form, by the Hodge decomposition theorem, then $\psi$ is a harmonic $0$-form for the kinetic energy operator. The addition of any potential term $\hat{V}(q)$ is necessarily either cohomologous with $\psi$, or must lie in a functional space orthogonal to $\psi$. The latter requires the wave function to be of the form $\Psi=\psi\oplus\phi$, with $\psi$ a 0-form and $\psi$ a 1-form. Let $F^0,F^1$ be the space of 0 and 1 forms on $T^\ast (\mc{L})$. Then $\Psi\in F^0\oplus F^1$, and $F^0\perp F^1$. The problem now becomes that the wave function $\Psi$ solves the equation $\braket{\Psi^\ast |\hat{H}|\Psi}=\lambda$, with $\lambda>0$ (here $\lambda$ is an eigenvalue). It cannot solve the equation that we intended it to solve ($i.e.~\lambda=0$), and any potential function $\hat{V}(q)$ must therefor be a vector in a space which is orthogonal to the free operator. 

Moreover, Hodge's orthogonality condition implies that $\hat{H}$ cannot even act on a two non-cohomologous functions originating from the same function space. Indeed, a wave function solution $\Psi$, must be comprised of the direct sum of two wave function in orthogonal spaces $(\Psi=\psi\large|_{\mc{L}^\second}\oplus\phi\large|_{\mc{L}'})$! Clearly, the rigid interpretation of the "potential" as a 1-form is not the intent behind Eq.~\eqr{hamdelta}. We will return to this point again in \secref{deltaprm}.

Specifying $\hat{V}=\delta$ as either a distribution or a measure, determines whether $\psi\in \mc{S}\left(\closure{\mbb{R}}\right)$ or $\psi\in \mc{B}\left(\Sigma\right)$. Clearly $\mc{S}\left(\closure{\mbb{R}}\right)\cap \mc{B}\left(\Sigma\right)\neq\emptyset$. In the former, the space of distributions is the space of continuous linear functionals $T:\mc{S}(\closure{\mbb{R}})\to\mbb{C}$ (or $\mbb{R}$). As a measure in the latter case, $\mu_{\delta}: \mc{B}\left(\Sigma\right)\to\mbb{C}$ (or $\mbb{R}$). However in quantum mechanics, we typically regard $\psi\in\mc{H}\subset L^p\left(\closure{\mbb{R}}\right)$. Since we impose a self-duality condition; $\psi^\ast=\psi$. Then the H\"{o}lder inequality formally restricts $\mc{H}\in L^2\left(\closure{\mbb{R}}\right)$. However $\mc{D}_{\delta_x}\not\subset L^p\left(\closure{\mbb{R}}\right)$ for any $1\leq p \leq\infty$, so $\mc{D}_{\delta_x}\nsubset \mc{H}$! Within this context, the question of the domain for Eq.~\eqr{hamdelta} ($\mc{D}_{\hat{H}}$) cannot be meaningfully addressed. The standard approach of extension parameters is of little help. Limiting sequences approaching the $\delta'$ (or $\delta$) distribution can be constructed with $L^2$ functions on compact subsets of $\mbb{R}$. However the limit point of the sequence inevitably has a domain which cannot be extended be $L^2(\mbb{R})$. Thus no self adjoint extension which is a result of limits of sequences exists, as they are not closed in $L^2$. 

In the above paragraphs, we outlined a number of inconsistencies with regard to Eqs.~\eqr{hamdelta} and \eqr{hamdelta2}. Neither one is a differential equation. Even worse, $\mc{D}_{\delta_x}\notin \mc{H}\subset L^2(\closure{\mbb{R}})$ and $\delta'_x$ is not even a measure. On a differentiable manifold, $\frac{\partial}{\partial x}\delta(x)$ defines 1-form, which implies that a free scalar solution and the interaction solution space are decoupled, and introduces an undefined inner product for any non-zero constants $\psi_0$. A similar point is raised in \cite[Secs. 3 and 4]{Park:1997fw}. We assume that Eq.~\eqr{hamdelta} originates from the variation of some linear functional (Lagrangian or Hamiltonian) density, $\mscr{H}$:
\beq\label{varH}
\frac{\delta}{\delta \psi^\ast}\int_{x\in\closure{\mbb{R}}}\mscr{H}(\psi^\ast,\psi)dx=\int_{x\in\closure{\mbb{R}}}\psi^\ast\left(\hat{H}\psi\right) dx&=\braket{\psi^\ast, \hat{H}\psi}-E\braket{\psi^\ast,\psi}=0,
\eeq
Obviously we require that $\psi$ be self-dual in Eq.~\eqr{varH}. We also require  that $\hat{H}$, at a minimum be essentially self adjoint, with bounded operator norm: $\norm{\braket{\psi^\ast,\hat{H}\psi}}_\infty<\infty$, with $\psi\in \mc{D}_{\delta_x}:\braket{\psi^\ast ,\hat{H}\psi}=\braket{\psi^\ast ,\hat{H}\psi}^\ast$. Quantum mechanics demands that Eq.~\eqr{varH} admit a definition on some dense subset of $L^2\left(\closure{\mbb{R}}\right)$ (or a suitably equivalent notion thereof), otherwise colloquially speaking, we break quantum mechanics. The above minimal requirements are met by assuming $\hat{H}$ to be of Schatten class. The semicontinuous topological vector spaces on $\closure{\mbb{R}}$ with Riemann(Lebesgue)-Stieltjes measure ($C_{L,R}$), which are defined in \secref{semi}, are locally compact Hausdorff spaces in $\mc{D}_{\hat{H}}$. That $\hat{H}$ is compact, trivially follows by the Alexandroff compactification for any wave function $\psi(x)\in C_{L,R}(\closure{\mbb{R}})$ such that $\hat{H}:C_{L,R}(\closure{\mbb{R}})\to\closure{\mbb{R}}$.     

We argue that the domain incompatibility, is inherently topological in nature. The $\delta_x$ functional requires continuity at the point of support for any function on which it acts. This immediately places the a solution in some subspace of $C\left(\closure{\mbb{R}}\right)$ for $\mc{D}_{\delta_x}$. Thus any solution should be in some subset of $C\left(\closure{\mbb{R}}\right)$ such that it admits a self adjoint extension to $L^2\left(\closure{\mbb{R}}\right)$. 

Singular differential systems have been well studied in terms of nonholonomic geometric mechanics. The works of \textit{Faddeev} and \textit{Vershik} (in particular \cite{faddeev199540,Vershik1988}) allow us to work with distributions on (co)tangent spaces, which exist as geometrical objects in their own right ($i.e.$ vectors as jets or germs of fields, and 1-forms as modules over jets or germs). The inherent differential structure of the (co)tangent spaces will be particularly useful in later sections, where Eq.\eqr{hamdelta2} will be defined on a configuration space endowed with a symplectomorphism structure. Below we construct a general formalism which aims to bridge this gap between the space of test functions for $\delta$ and the free kinetic energy operators more satisfactorily.

There are many approaches available in order to tame singular Hamiltonian (and Lagrangian) systems. There is also an overwhelming number of papers which apply functional calculus methods to point-supported interactions of Schr\"{o}dinger operators. A small subset of these cover Green's Function methods, deficiency indices, and extensions to $L^2$ or Sobolev spaces (dense subspaces of $L^2$)~\cite{Park:1996, Park:1997fw, albeverio1998symmetries, albeverio2002point, Nizhnik2006, gill2008, Eckhardt2011SL, Lange2015, Gadella2016}. The seminal works by Albeverio \textit{et al}, spanning over three decades, is summarized in~\cite[and refs therein]{albeverio2012solvable}, and worth particular mention due to the multitude of systems analyzed using the method of self adjoint extensions of symmetric operators. In particular, the study of propagators of quantum mechanical Hamiltonians with regular, and singular potentials in many spatial dimensions. Approximation methods determine estimates for boundedness and well-posedness in finite difference Schr\"{o}dinger equation (as well as the NLS, and semi-relativistic variants). See \cite{albeverio2012solvable,albeverio2000singular,TakFaddSpect2015,StrichBfieldDAncona2009,LewinSabinHartee2015,StrichineqFLL2013}, and references therein for instance.

 If one looks to non-perturbative methods for handling point supported interactions such as Eq.~\ref{hamdelta2}, nonlinear distributional solutions are invariably the only other tool at one's disposal. On the other hand, there has been a tremendous amount of work done in the fields of nonlinear functional analysis, with special attention to point interactions. Colombeau algebras are a considerably intricate and abstract formalism which have had some success in recent years with the construction of generalized functional algebras. The difficulty inherent in Colombeau algebras is matched only by their potential for use in large classes of function spaces. For works on the general theory of Colombeau algebras see ~\cite{colombeau1990, oberguggenberger1992, Jelinek1999}. An interesting exposition on a modern generalization which simplifies some of the formalism, and discusses the current challenges of Colombeau algebras is~\cite{Nigsch2013}. Recently~\cite{SchwarzOps} has appeared, which outlines a generalization for algebras of operators and distributions. 
 
Here we will not need to employ such generalized formalisms, though there is certainly some overlap with the afore mentioned in all cases. The approach here is a construction from first principals, in terms of functional methods on topological vector spaces with particular measure properties, differentiable manifolds\cite{faddeev199540,Vershik1988,TakFaddSpect2015}, and the spaces of integrable distributions~\cite{Talvila2009,TalvilaLp,Talvilafinord}. 

It is worth making a particular mention of the works of Johnson and Lapidus~\cite[and refs therein]{johnlap}, which became known to the author only after the completion of the initial draft of this work. The work here contains some parallels of Johnson and Lapidus~\cite{johnlap} in terms of the usage of the L-S measures in Feynman's path integral. However, in this work, we build Borel measurable and Banach spaces on the foundation of half-open Borel generating sets on $\closure{\mbb{R}}\cong S^1$. In this setting, L-S measures are in some sense, very natural measures for such semicontinuous Banach spaces. Specifically, L-S measures have been precisely chosen to coincide with the generating sets of the underlying topological vector space (TVS). This has certain benefits in the analysis below, and which are not generally possible in the standard Hilbert (or Banach) spaces. For instance, we have a topology compatible with notions of making identifications of certain (even singular) distributions with the derivatives of their primitive distributions, in particular see \cite{Talvila2009}. The ability to make these identifications, offers an intriguing option which may potentially (if generalizable in a meaningful way) expand the tools available to define and evaluate Feynman integrals including those with singular measure potentials. Another benefit of our construction is that many operators are naturally self adjoint on the semicontinuous manifold spaces defined below.
 
 \subsection{Gauge integrals, primitive functions, and Krein spaces}\label{ghk}
 We would like to have a mapping $\mc{I}$ such that if $\hat{H}\in\mc{L}^\ast$ where $I:\mc{L}^\ast \to\mc{L}'^\ast\to\mc{L}''^\ast$, where $\ast$ denotes the dual space of continuous linear functionals. In this case then it is at least, in principle, possible to have some $\mc{D}_{D^2}=\mc{D}_{\delta'}$. It is well known that the generalizations of the Riemann and Lebesgue integrals are the class of gauge integrals (and in particular the Henstock-Kurzweil (HK) and Stieltjes classes,\cite{Bartle1996,abbott2010,Bongiornominsolns,Bongiornolebesgue,Bongiornoinfvar}), which are known to integrate functions which are derivatives of unique primitive functions. 
  
Of particular relevance to our discussion here are the regulated classes of gauge integrals with Lebesgue-Stieltjes measure \cite{Talvila2009,TalvilaLp,Talvilafinord}, as well as Krein function spaces\cite{Kreinlinear,Kreintopics,Rovnyak2002,ellis2003}. The gauge integrals define a gauge function within subintervals $I\subset \mbb{R}$, and a tagged partition of $I$. The uniqueness of primitive functions obtained from integrable functions and distributions affords one the luxury of straight forward identifications of domains of certain classes of functionals (particularly the space of Schwarz functions over some topological metric space, $\mathscr{T}(X)$) where the inversion of distributional derivatives is possible. These are the classes of integrable distributions. 

A function $\gamma$ which maps some interval $[a,b]$ to $\mathbb{R}$ is called a {\it gauge} on $[a,b]$ if $\gamma (x)>0$ for all $x\in [a,b]$. A {\it tagged partition} is a finite set of pairs of closed intervals and tag points in $\closure{\mathbb{R}}$, $\mathcal{P}=\{([x_{n-1},x_n],\tilde{x}_n)\}^N_{n=1}$ for some $n\in\mathbb{N}$, with $\tilde{x}_n\in[x_{n-1},x_n]$ for each $1\leq n\leq \mathbb{N}$ and $-\infty=x_0<x_1<x_2\ldots<x_N=\infty$. The finite pair set consisting of a tagged partition and tag points is denoted by $(\mathcal{P},\{\tilde{x}\}^N_{n=1})$. For a particular $\gamma(x)$, a tagged partition $(\mathcal{P},\{\tilde{x}\}^N_{n=1})$ is said to be $\gamma$-{\it fine} if every subinterval $[x_{n-1},x_n]$ satisfies $x_n-x_{n-1}<\gamma(x)$. Therefore a gauge $\gamma$ on $\closure{\mathbb{R}}$ together with a tagged partition, maps to open intervals in $\closure{\mathbb{R}}$ and for each $x$ in some subinterval $[a,b]\subset \closure{\mathbb{R}}$, then $\gamma(x)$ is an open interval containing $x$. 

In particular we will generally consider normed linear topological metric spaces generated by the collection of all half-open Borel sets over the extended real line $\closure{\mathbb{R}}=[-\infty,\infty]$, which is also one possible way to define a compactification of $\mathbb{R}$. Krein spaces are useful due to their structure, since they can include regular subspaces which are isomorphic to Hilbert spaces of quantum mechanics. Krein spaces will evolve naturally out of the formalism developed here. 
 
 The organization of the paper is the following. In~\secref{semi}, we introduce and define the spaces of semicontinuous functions. We begin with a discussion of some idiosyncrasies regarding certain definitions of specific tempered distributions. The Heaviside distribution and general classes of step functions are discussed in detail. These discussions serve as the motivations which follow in the latter sections of~\secref{semi}, where we define the measure spaces of half-open Borel topologies, the spaces of continuous measurable (and therefore bounded) functions $C(\closure{\mbb{R}})$ and the isomorphic left/right semicontinuous topological spaces $C_{L,R}(\closure{\mbb{R}})$ defined over the half-open Borel set topologies. We then discuss the $L^p$ analogs to the semicontinuous spaces $C_{L,R}$, which are defined for non-atomic semicontinuous $L^p$ functions, excluding sets (and collections of sets) of Lebesgue measure zero, such as fat Cantor sets or Cantor-Lebesgue measure. With a refined notion of the standard $L^p$ equivalence class identifications, we construct the semicontinuous quotient spaces $L^p_{L,R}$ which are continuous with respect to the half-open Borel measure space topologies for $1\leq p\leq\infty$. Under these equivalence class identifications, we achieve a continuous and dense partial embedding of non-atomic $L^p_{L,R}$ functions into the spaces of $C_{L,R}$. This considerably enlarges the classes of functions in $C_{L,R}$ which are topologically continuous with respect to the base Borel topologies. Theorems are proved regarding the topologically continuous dual spaces of all of the defined function spaces. The topological and norm closure of the dual spaces of semicontinuous functions over half-open Borel topologies with the $\norm{\cdot}_{\sup}$ is given by the function spaces of semicontinuous bounded variation $\mc{BV}_{L,R}$ with $\norm{\cdot}_{\mc{BV}_{L,R}}$ which are finitely additive over all collections of compact subsets of $\closure{\mbb{R}}$. The benefit of this construction is that with the half-open L-Sj measures, there is a one-to-one mapping of $L^\infty_{L,R}$ functions to the semicontinuous spaces of bounded variations, which can be generalized to the Riemann-Stieltjes measures for gauge integrals. The Radon-Nikodym theorem provides a description of the semicontinuous spaces of absolutely continuous $(AC_{L,R})$ functions in terms of the second fundamental theorem of calculus, which is reflexive and includes primitives for semicontinuous integrable distributions. The spaces of semicontinuous integrable distributions are the finitely additive measures of $\mc{BV}_{L,R}$ with the corresponding $\norm{\cdot}_{\mc{BV}_{L,R}}$. We note that by continuity in these spaces, the $\norm{\cdot}_{L,R;\sup}$ (equivalent to the $\norm{\cdot}_{L^1_{L,R}}$) bounds all $\norm{\cdot}_{L^p_{L,R}}$, for all $p$. In this manner we have containment of the norms 
 $\closure{\norm{\cdot}}_{L^p_{L,R}}\subseteq\closure{\norm{\cdot}}_{L,R;\sup}=\norm{\cdot}_{L,R;\sup}\cup\norm{\cdot}_{\mc{BV}_{L,R}}$, with equality in the case of $\norm{\cdot}_{L^1_{L,R}}$. 

 In~\secref{deltaprm} we will discuss the Hamiltonian in terms of differential geometric structures and utilize the semicontinuous spaces of functions to analyze the Hamiltonian functional equation Eq.~\eqr{hamdelta2}. Placing the Hamiltonian on a differential manifold will yield a geometric approach, which will give further support to the functional approach that we are proposing. In terms of the semi-continuous topological spaces, we will find the corresponding Hilbert space $\mc{H}$ for the Dirac-$\delta'$ system, which is separable, and admits a semicontinuous orthogonal decomposition such that $\mc{H}=\mc{H}_L\oplus\mc{H}_R$. Therefore $\mc{H}$ is measure valued projective space. The Hilbert space is defined as Sobolev spaces of semicontinuous functions such that the Hamiltonian is bounded in the operator norm. Essentially $\mc{D}({\hat{H}}$ is the semicontinuous functions with $\hat{H}$ Schatten class. These are simply $\left(L^2(\closure{\mbb{R}})\cap C^\infty_{L,R;0}(\closure{\mbb{R}})\right)\subset\closure{C_{L,R}(\closure{\mbb{R}})}$, where the bar denotes the set closure of $C_{L,R}(\closure{\mbb{R}})$.
  
 In~\secref{Krein} we summarize the structure and properties of the indefinite Krien spaces denoted by $\mf{H}$, which contains subspaces that are isomorphic to the Hilbert space $\mc{H}$. It is shown that the Hilbert space and its associated negative norm antispace correspond to the sign of the coupling term to the $\delta'$ potential. The Hilbert space topology is the strong topology of $\mf{H}$, and therefore $\mf{H}$ inherits the orthogonal decomposition from the Hilbert space/antispace states in addition to the orthogonal decomposition in terms of semicontinuous functions from $\mc{H}$. We close the paper with a short summary and concluding remarks on future works in progress.
  
\section{Spaces of Semi-Continuous Functions}\label{semi}
Many points discussed in the previous section will become relevant if we consider the Riemann(Lebesgue)-Stieltjes integral of semi-continuous functions/distributions with measures defined by the Borel sets of half open intervals over $\closure{\mbb{R}}$. We want our space to be reflexive so that the weak$^\ast$ equals the norm topology. Our strategy will be to use the half open topologies on $\closure{\mbb{R}}$, along with set inclusion/exclusion definitions in order to define a measure. Let $\left(\closure{\mbb{R}},\mathscr{B},\mu\right)$ denote our measure space over the extended real line. We also define a normed linear space, $\mc{L}\left(\closure{\mbb{R}},\norm{\cdot}\right):=\mc{L}\left(\closure{\mbb{R}}, \exists~\mu (f)~|~for~some~f,~\norm{f}\in\mc{B} \right)$. Thus we will have a Banach space $\mc{B}$ over $\closure{\mbb{R}}$. Then Riemann-Stieltjes integration will be semicontinuous with respect to the norm inherited through the weak$^\ast$ topology on $\mc{L}$.

In particular we will have a Banach space with isometric isomorphic dual over the space of $C_0^\infty (\closure{\mbb{R}})\subsetneq L^p$-space. It is well known that $L^p$ is the completion of $C_0^\infty (\closure{\mbb{R}})$ with respect to the $L^p$-norm. We will show that with restriction to the solution space of  Eq.~\eqr{hamdelta2} that we may extend our space to a subspace of $L^p(\closure{\mbb{R}},d)$ for $p=2$. This extension will be necessary in order that the domain of kinetic energy operator and the space of test functions for the potential agree. It is also well known the domain of the free Hamiltonian operator $\mc{D}_{H_f}$, is essentially self adjoint on $L^2(\closure{\mbb{R}})$. Let $k$ be Fourier conjugate variable to $x$, then the unique self adjoint extension is the subspace of $L^2(\closure{\mbb{R}})$ functions with Fourier transforms quadratic in $k$. Therefore we need to show that $\mc{D}_H\subseteq \mc{D}_{H_f}$ and that $\mc{D}_H$ is self adjoint on its domain.

Our approach will rely heavily on arguments of continuity, both algebraic and topological. We first start with the class of step functions $\chi$ over finite intervals, which are dense in $L^p$ for all $p$. This is a natural bridge between $L^p$ and $C(\mbb{R})$, as the set $C_0(\closure{\mbb{R}})$ is dense $C(\closure{\mbb{R}})$, and $L^p(\closure{\mbb{R}},\|\cdot\|_p)=\mc{C}(C(\closure{\mbb{R}},\|\cdot\|_{\sup}))$, the norm closure of $C(\closure{\mbb{R}},\|\cdot\|_{\sup})$. Thus any $f\in C(\closure{\mbb{R}})$ can be approximated as some sequence of step functions, $\{\chi_n\}\in L^p(\closure{\mbb{R}})$, and therefore Lebesgue (or gauge Lebesgue-Stieltjes, HK-Stieltjes) integration in $L^p$ is sequentially equivalent to Riemann-Stieltjes integration in $C(\closure{\mbb{R}})$. 

\subsection{Spaces of semi-continuous functions}
 Consider the semi-continuous step functions defined from subsets of the collection of all Borel generating sets of half open intervals in $\closure{\mbb{R}}^1$. The generalization to $\closure{\mbb{R}}^n$ is straight forward. For example, the Heaviside distribution (Fig.\ref{fig1}) can be uniquely defined as a semi-continuous function both from the left $(H_L=H_{(0,\infty]})$ and from the right $(H_R=H_{[0,\infty)})$. \textbf{\textit{However}} uniqueness is lost with the Lebesgue measure. 
 
\begin{figure}
\centering
\begin{tabular}{p{6cm} p{6cm}}
\begin{subfigure}[h]{0.3\textwidth}
\begin{tikzpicture}
	\begin{axis}[my style, scale=.5, ylabel={$H_R(x)$}, xmin=-2, xmax=2, ymin=-2, ymax=2]
	
	\addplot[domain=-3:0]{0};
	\addplot[domain=3:0]{1};
	
	\addplot[mark=*, fill=white]  coordinates {(0,0)};
	\addplot[mark=*] coordinates {(0,1)};
	\end{axis}
\end{tikzpicture}
\subcaption{Right continuous Heaviside}
\label{rthvsd}
\end{subfigure}
&
\begin{subfigure}[h]{0.3\textwidth}
\begin{tikzpicture}
	\begin{axis}[my style, scale=.5, ylabel={$H_L(x)$}, xmin=-2, xmax=2, ymin=-2, ymax=2]
	
	\addplot[domain=-3:0]{0};
	\addplot[domain=3:0]{1};
	
	\addplot[mark=*, fill=white]  coordinates {(0,1)};
	\addplot[mark=*] coordinates {(0,0)};
	\end{axis}
\end{tikzpicture}
\subcaption{Left continuous Heaviside}
\label{lfthvsd}
\end{subfigure}
\end{tabular}
\caption{One-Sided Continuous Heaviside Functions}\label{fig1}
\end{figure}
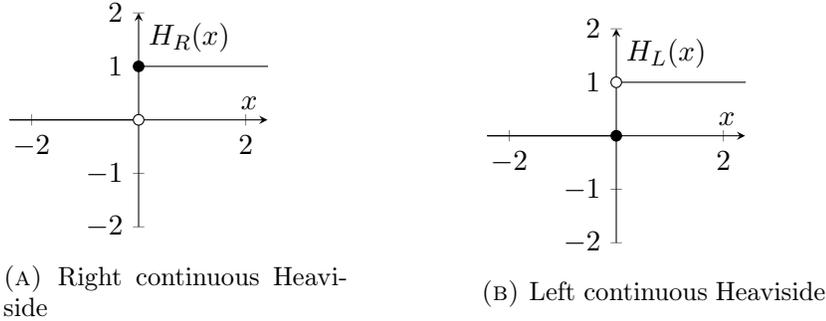

A natural question one may ask is, why are half-open topologies necessarily helpful? One benefit is that in a sense, uniqueness is gained. Here, there is only one regulated, semi-continuous Heaviside function on each measurable space $\left(\closure{\mbb{R}},\mathscr{B}_{(\cdot,\cdot]}\right)$, and $\left(\closure{\mbb{R}},\mathscr{B}_{[\cdot,\cdot)}\right)$. With respect to the corresponding defining topologies, $H_L$ and $H_R$ are unique topologically continuous $L^\infty$ functions. 

The above sense of uniqueness, along with the function spaces to be defined in the following sections, will admit maps from $\mc{S}^\prime$ to spaces of regular distributions for distributions that have primitive functions, such as $\delta$. Such mappings, when they exist, permit well defined functional notions of integration by parts and exponentiation. This would be of use for physicists who work with Feynman's functional integral, which motivates the particular choice of application discussed in~\secref{deltaprm}. In essence, the functional tools at ones disposal, is enlarged with respect to such classes of distributions. Furthermore, the $S^1$ topology utilized herein, is naturally compatible with world sheet topologies of closed (super)strings and we posit the potential existence for non-trivial applications there as well.    

In order to better motivate our later definitions and conventions, we first discuss a trivial example which still highlights particular nuances in semicontinuous spaces. Consider the semi-continuous Heaviside distributions in Fig.\ref{fig1}. Note that $H_L$ and $H_R$ are regulated in the sense of~\cite{Talvila2009}. When paired with some $\phi\in\mc{S}$, these meet the standard definition of a distribution on $\mc{S}$ such that $T_{H_{L,R}}:=\int_{\mbb{R}}{}H_{L,R}(x)\phi(x)dx$, or rather $T_{H_{L,R}}:\mc{S}\to\mbb{R}^+$. They define a semi-positive definite mapping from the space of Schwartz functions to the field of scalars $\mbb{R}^+$. Now consider the signum (sgn) distribution. One typically encounters various definitions in textbooks. For example

\begin{equation}\label{sgn1}
\begin{aligned}
\sgn(x):=\begin{cases}
~1, & x>0\\
-1, & x<0
\end{cases}
\end{aligned}
\end{equation}
or the semicontinuous variants. Another way to define the sgn distribution is to use the distributional identity defined by a linear combination of the completely discontinuous Heaviside distribution\footnote{Though this distributional identity does not hold pointwise with respect to arbitrary measure.}: $H(x)=1$ if $x>0$ and $H(x)=0$ if $x<0$,
\begin{align}\label{sgn2}
\sgn(x)=H(x)-H(-x).
\end{align}
Indeed, take $\phi(x)\in\mc{S}$, which necessarily implies $\phi'(x)\in\mc{S}$. With Lebesgue measure and the usual topology on $\mbb{R}$, then
\beq\label{sgnHid}
\braket{\sgn,\phi'}&=\int_{\mbb{R}}\sgn(x)\phi'(x)dx\\
&=\int^\infty_0 \phi'(x)dx+\int^0_{-\infty}(-1)\phi'(x)dx\\
&=\braket{H(x)-H(-x),\phi'(x)}\\
&=-2\phi(0)
\eeq
The equivalence definitions of sgn in Eqs.\eqr{sgn1} and \eqr{sgn2} is a distributional identity only. However two sequences, $\{g_n\},\{f_m\}$ say, converge to the same distribution $(T_g\leftrightarrow T_f)$ if and only if they converge pointwise in the dual topology. In the case of the semi-continuous Heaviside distributions $H_L$ or $H_R$ (Fig.\ref{fig1}) in Eq.\eqr{sgn2} with a Stieltjes measure on the half-open measure topology, the above distributional identity in Eq.\eqr{sgnHid} does not hold.

This can be seen in the following (see also~\cite{Talvila2009,TalvilaLp}). Take $\phi(x)\in\mc{S}$. From point reflection $(x\to-x)$ and the definition of $H_L(x)$ in Fig.\ref{fig1} we find 
\beq\label{HL}
H_L(-x)=\begin{cases}
1, & x<0\\
0, & x\geq 0
\end{cases}
\eeq
Then for sgn$(x)$ as in Eq.\eqr{sgn2}, but using $H_L(x)$ and a left semicontinuous L-S measure instead (here $\lambda_L(dx)\sim dx$), we have the distributional result
\beq\label{Lsgn1}
\braket{\sgn(x),\phi'(x)}&=\int_{\closure{\mbb{R}}}\sgn(x)\phi'(x)dx\\
&=\int^{0}_{-\infty}\left(H_L(x)-H_L(-x)\right)\phi'(x)dx+\lim_{\epsilon\to0^+}\int_{\epsilon}^{\infty}\left(H_L(x)-H_L(-x)\right)\phi'(x)dx\\
&=\int^{0}_{-\infty}-H_L(-x)\phi'(x)dx+\lim_{\epsilon\to0^+}\int_{\epsilon}^\infty H_L(x)\phi'(x)dx\\
&=\left(0\cdot-\phi(0)+H_L(-\infty)\phi(-\infty)\right)+\phi(\infty)-\lim_{\epsilon\to 0^+}\phi(\epsilon)\\
&=\lim_{\epsilon\to 0^+}\int^{\infty}_{\epsilon}H_L(x)\phi^\prime (x)dx\\
&=\braket{H_L(x),\phi^\prime(x)}
\eeq
In this case we have the distributional identity of $\sgn(x)=H_L(x)$. Analogous calculations yield the same result using $H_R(x)$ or $H_L(x)$ in Eq.~\eqr{sgn2} and the corresponding right or left semi-continuous half-open topologies\footnote{The half-open topologies here are measure norm topologies in terms of Lebesgue-Stieltjes measures.}. An analogous calculation to \eqr{Lsgn1} using the standard Lebesgue measure instead of the L-S measure, one obtains zero! The topology\footnote{We say "topology" here because the L-S measure is chosen to be continuous with respect to the defining topology.} is clearly important regarding distributional identities, and sets of measure zero are now relevant.

We may view the discrepancies between the last example and the distributional "identity" Eq.~\eqr{sgn2} resulting from the lack of reflection symmetry with the distributions $H_L$ and $H_R$\footnote{This does not occur in definitions where there is a discontinuity from both the left and the right directions, as in Eq.\eqr{sgn1}. This is a point which we will return to later.}. Under reflections $x\to-x$, $H_L(x)$ and $H_R(x)$ do not maintain the direction of semi-continuity. Moreover, if we were to employ the distributional derivative after the first line in Eq.\eqr{Lsgn1}, we would incorrectly conclude that $\braket{\sgn(x),\phi'(x)}=-2\braket{\delta(x),\phi(x)}$, in agreement with the expected distributional derivative of sgn. However, doing so after the last line in Eq.\eqr{Lsgn1}, we arrive at $\braket{\sgn(x),\phi'(x)}=-\braket{\delta(x),\phi(x)}$. 

\subsubsection{A glance at De Rham cohomology via homotopies}\label{deRham}
In the space of distributions (or \textit{De Rham} cohomology), the two distributions are equivalent since they differ by a constant, however a homotopy analysis shows that they produce distinct Euler characters. This gives us a hint regarding the nature of the discontinuity. $H_L(x)$ and $H_R(x)$ are each discontinuous at one point. However, using the standard Lebesgue measure, both fail to produce a non-zero distributional identity. 

With Ex.\eqr{Lsgn1}, one may define a contractable homotopy map, such that it may be a representative of the 0-$th$ \textit{De Rham} class $H^0(\closure{\mbb{R}})$, of the piecewise semicontinuous 0-form $H_L(x):~dimH^0(\closure{\mbb{R}})=dim~(\closure{\mbb{R}})=1$, where $b_0$ is the Betti number. Define the homotopy parameter $t$, such that it is locally equivalent to the directional (left/right) limit of the singular point in $H_L(x)$ and $\delta(x)$. It follows that the Euler character (with $\delta(x)dx$, a 1-from) is easily seen to be $\chi(\closure{\mbb{R}})=(-1)^0 (1)+(-1)^1 (1)=0=\chi(S^1)$. In this case we have that the homotopies $\sgn(x)$ and $H_L(x)$ are cohomologous, which implies that $\mathbold{d}(H_L(x)-H_L(-x))=\delta(x)\ dx$ is a $C^\infty(\closure{\mbb{R}},\mbb{R}^+)$ projective diffeomorphism on $\closure{\mbb{R}}^1$, in the sense that the support of $\mathbold{d}(H_L(x)-H_L(-x))\in (0^+,\infty]$ and not $(-\infty,\infty]$.

We repeat the analogous calculation with $\sgn(x)$ (Eq.~\eqr{sgn1} on $\closure{\mbb{R}}$) with the usual topology and Lebesgue measure. The Euler character is $\chi(\mbb{\closure{R}})=(-1)^0 (2)+(-1)^1(1)=1$. $\sgn{x}$ maps $\closure{\mbb{R}}$ into two disconnected components $(\mbb{R}^-,\mbb{R}^+)$. Here, $\mathbold{d}\sgn(x)=2\delta(x)\ dx$, which is no longer a projection, as was the case above. As a homotopy, $\mathbold{d}\sgn(x)$ must either map to inequivalent De Rham groups depending on which measure topology one implements, or must violate the equivalence between the homotopy and De Rham groups. No matter the case, the distinct Euler characters show that there is an inherent topological difference between $\sgn$ (as Eq.~\eqr{sgn1}) and linear combinations of $H_L$ (as in Ex.~\ref{Lsgn1}), and analogously for $\sgn$ defined from $H_R$. 

$\closure{\mbb{R}}^1$ is homeomorphic to $S^1$. It is well known that for $S^1$, the two cohomology groups of compact support are equivalent, $H^0_c(S^1)=H^1_c(S^1)=\mbb{R}$. Moreover, generally for some compact manifold $M$, $H^p_c(M)=H^p(M)$. In terms of the above homotopies, a discontinuity at the origin of $\closure{\mbb{R}}$ is equivalent to a cut at a some point on $S^1$. For semicontinuous topologies, a discontinuity depends on the direction (orientation) in which the limit is taken, or rather, on the left/right half-open interval topology. The semicontinuous homotopy projects onto the continuous path connected component of $\closure{\mbb{R}}$. 

We see that we may identify the above topological distinction in the half-open topologies with a violation of reflection symmetry, if we choose the convention to define the sgn distribution(s) such that the reflection symmetry is maintained. We therefore define the left semi-continuous sgn function as
\beq\label{sgnL}
\sgn_L(x):=H_L(x)-H_R(-x),
\eeq
and similarly, the right semi-continuous sgn function as 
\beq\label{sgnR}
\sgn_R(x):=H_R(x)-H_L(-x).
\eeq

\begin{example}{2.1.1:}\label{expl3}
 With respect to the left continuous half-open measure topology  (here, $\mu_L(dx)\sim dx$)\footnote{We could have instead used the left semicontinuous L-S measure: $\lambda_L(dx)\sim dx$. We will establish these equivalence classes below.}, $\sgn_L$ yields the expected distributional equivalence,
 \beq\label{sgnL}
 \braket{\sgn_L(x),\phi'(x)}_{(-\infty,\infty]}&=\int^0_{-\infty}\sgn_L(x)\phi'(x)dx+\lim_{\epsilon\to 0^+}\int^{\infty}_{\epsilon}\sgn_L(x)\phi'(x)dx\\
 &=\int^0_{-\infty}\left(H_L(x)-H_R(-x)\right)\phi'(x)dx+\lim_{\epsilon\to 0^+}\int^{\infty}_{\epsilon}\left(H_L(x)-H_R(-x)\right)\phi'(x)dx\\
 &=\left(0\cdot\phi(0)-0\cdot\phi(-\infty)\right)-\left(\phi(0)-\phi(-\infty)\right)+\\
 &\qquad\left(\phi(\infty)-\lim_{\epsilon\to 0^+}\phi(\epsilon)\right)-\left(0\cdot\phi(\infty)-0\cdot\lim_{\epsilon\to 0^+}\phi(\epsilon)\right)\\
 &=-\phi(0)-\phi(0^+)\\
 &=-2\phi(0).
 \eeq
Analogous calculations show,
\begin{align}
\label{sgnR,oc}
\braket{\sgn_R(x),\phi'(x)}_{(-\infty,\infty]}&=0\\
\label{sgnL,co}
\braket{\sgn_L(x),\phi'(x)}_{[-\infty,\infty)}&=0\\
\label{sgnR,co}
\braket{\sgn_R(x),\phi'(x)}_{[-\infty,\infty)}&=-2\phi(0).
\end{align}
The measure and measure space topologies in Eqs.\eqr{sgnL}-\eqr{sgnR,co} were chosen to coincide with the semi-continuity of $H_L$ and $H_R$, which has obvious generalizations to the entire class of step functions.
\end{example}

Another homotopy analysis of the last example, shows that $\mathbold{d}\sgn(x)$ is still a projective diffeomorphism such that its support is in either $(-\infty,\infty]$ or $[-\infty,\infty)$, depending on the particular chosen left/right half-open topologies. Furthermore, for each non-zero result in Ex.~\ref{expl3}, the Euler character is $\chi(\closure{\mbb{R}})=0$, as we would like. Hence we confirm that reflections of semicontinuous homotopic maps is a homeomorphism invariant of the Euler character. 

Ex.~\ref{expl3} implies that weak equivalences with respect to classes of step functions may be engineered such that reflection symmetry is maintained or broken, depending on one's particular preference. In what follows we will study topological spaces which preserve the reflection symmetry of the classes of semicontinuous step functions. In this sense, the classes semicontinuous of step functions are the simple function representatives of the maximally symmetric classes of functions of these spaces. 

\subsubsection{Spaces of semicontinuous functions in $C\left(\closure{\mbb{R}}\right)$}\label{CR}
With the above in mind, we make the following definitions. 
\begin{defin}\label{measurespaces}
Let $\mathscr{B}$ be the collection of all generating Borel sets on $\closure{\mbb{R}}$, such that $\mathscr{B}$ is a $\sigma$-algebra with the usual topology. We denote by $\mathscr{B}_L$, all countable disjoint unions of left continuous half-open Borel sets $(\cdot,\cdot]$, taken as the generating sets for $\mathscr{B}_L$. Similarly we denote by $\mathscr{B}_R$, the half-open right continuous generating Borel sets $[\cdot,\cdot)$. Clearly $\mathscr{B}_L,\mathscr{B}_R\subseteq\mathscr{B}$ are also $\sigma$-algebras on $\closure{\mbb{R}}$. We define the $\sigma$-finite measure space $X=\left(\closure{\mbb{R}},\mathscr{B}, \lambda\right)$, with $\lambda$, the standard Lebesgue measure on $\closure{\mbb{R}}$. We also separately define the measure spaces for the left and right half-open Borel sets as $X_{\mu_L}=\left(\closure{\mbb{R}},\mathscr{B}_L,\mu_L\right)$ and $X_{\mu_R}=\left(\closure{\mbb{R}},\mathscr{B}_R,\mu_R\right)$ respectively, where $\mu_L=\mu_{(\cdot,\cdot]}$ and $\mu_R=\mu_{[\cdot,\cdot)}$ are the appropriate  ($\sup$, $\braket{\cdot,\cdot}$) norms. 
\end{defin}
\textbf{Note:} Since we are on $\closure{\mbb{R}}\cong S^1$, half-open intervals are indeed open sets \cite[Ch. 6.3 a]{roman1}. Thus we have no problems taking the left (resp. right) half-open intervals as generating sets. It is also important to point out here that $X_{\mu_L,\mu_R}$ are \textbf{\textit{not}} considered \textit{a priori} to be \textit{bitopological} spaces. Though, it is surely possible to define such set structures. $X_{\mu_L,\mu_R}$ is notational convenience to collectively denote the distinct measure spaces $X_{\mu_L}$ and $X_{\mu_R}$. 

\begin{defin}\label{measure}
Let $\lambda$ be the standard Lebesgue measure on $\closure{\mbb{R}}$ generated by any Borel set $\mathscr{B}$. We define the measure inclusions for $\mathscr{B}_L$ and $\mathscr{B}_R$ to be that $\mu_L\preceq\lambda$, and $\mu_R\preceq\lambda$, such that $\mu_L,\mu_R$ will be the finer or equivalent continuous topologies with respect to the Lebesgue measure $\lambda$. For Lebesgue measure zero sets (LMZ sets), we regard all topologies and measure spaces $X,X_{\mu_L},X_{\mu_R}$ as equivalent.  
\end{defin}

Obviously all topologies generated on $\mbb{R}$, the usual topologies ($i.e.$ half-open topologies or other), are homeomorphic and generate paracompact subsets with respect to $\lambda$. Using generating Borel sets, we have that for any finite half-open $Y_{\mu_L,\mu_R}\subset X_{\mu_L,\mu_R}$, then $Y_{\mu_L,\mu_R}$ will also contain open, closed, and half-open subsets which are finer to $Y_{\mu_L,\mu_R}$. Similarly, for some finite $Y\subset X$, $Y$ will contain finer closed, open, and half-open subsets. Thus for any continuous mapping $f$, between the topological metric spaces $X,X_{\mu_L,\mu_R}$ with the respective norm topologies, we assume that generally there is always a coarser or finer gauge refinement such that $f^{-1}$ exists and is also continuous. We assume the relative topologies to be the weakest relative topologies such that $f$ will be bijectively continuous. Or rather, for any Lebesgue measure $\lambda$, we may define $\mu_L,\mu_R$ such that either may be extended by an appropriate set of measure zero (with respect to the half-open topologies) giving $\closure{\mu}_L,\closure{\mu}_R=\lambda$. The notable characteristic of this construction is that sets which contain finite jump discontinuities with respect to $\lambda$ are defined such that a finite discontinuity becomes left/right semicontinuous. In this respect, we regard $\mu_L,\mu_R$ as continuous Lebesgue-Stieltjes measures on subsets of $X$, and $\lambda$ as a continuous Lebesgue measure on subsets of $X_{\mu_L,\mu_R}$. Thus for LMZ sets, $\closure{\mu}_L(a)=\closure{\mu}_R(a)=\lambda(a)$, for some $a\in X$, such that the LMZ set $\{a\}$ is closed in all spaces, $X,X_{\mu_L,\mu_R}$. 

It is well known that measure theory of functions and topology can be intimately linked\cite{roman2}. The advantage of the coinciding half-open Borel topologies and measure topologies, is that we have the ability to use the discrete reflection (or rather parity) symmetry to describe a (partial) continuous symmetry on the subspaces $X_{\mu_L,\mu_R}$ separately. 

\begin{lem}\label{lemma1}
The topological metric spaces $X_{\mu_L}$ and $X_{\mu_R}$ are isometric homeomorphisms under the continuous identity mapping $e$, such that $e: X\cong X_L$ and $e:X\cong X_R$. Moreover, the topological metric subspaces $X_L\cong X_R$ are isometrically homeomorphic under $e$.
\end{lem}
\begin{proof}{:}
Trivial. $e$ is a bicontinuous mapping of open (closed) subsets of the topological space $\left(\mbb{R},\mathscr{B}\right)$ onto the spaces $\left(\mbb{R},\mathscr{B}_L\right)$ and $\left(\mbb{R},\mathscr{B}_R\right)$. $e$ is a distance preserving map in the respective norm topologies, and therefore $X\cong X_{\mu_L}$ and $X\cong X_{\mu_R}$ are isometric homeomorphisms. Transitively, $X_{\mu_L}\cong X_{\mu_R}$.
\end{proof}

\begin{rmk}\label{rmrk32}
Clearly $X$ is a homogeneous space under the identity map. We could have chosen not to separately define the spaces $X_{\mu_L},X_{\mu_R}$, then trivially shown that they are isometrically homeomorphic. However it in what follows it is better to have these spaces separately defined. Indeed as we have seen in Ex.\ref{expl3} above, the class of step functions do not possess equal measures between $X_{\mu_L}$ and $X_{\mu_R}$.   
\end{rmk}

At this point we have equivalence between the topological metric spaces under the identity map in the respective norm topologies. However in light of Ex.\ref{expl3} equivalence between linear function(al) spaces is clearly not possible for all function spaces. In particular for the class of semicontinuous step functions, reflection symmetry (for the semicontinuous sgn functions\footnote{as defined from Eqs.\eqr{sgnL} and \eqr{sgnR}}, and thus linear combinations of $H_L$ and $H_R$) and semi-continuity are only isometrically preserved in the spaces of corresponding semicontinuous norm topology.

\begin{lem}\label{lemCR}
Let $C\left(\closure{\mbb{R}}\right)$ denote the class of bounded, linear, continuous functions over $\closure{\mbb{R}}$. For any function $f$ defined on subsets $C\left(\closure{\mbb{R}}\right)$, the normed linear vector spaces $C_{\mu_L}\equiv C\left(\closure{\mbb{R}},\mu_L\right),C_{\mu_R}\equiv C\left(\closure{\mbb{R}},\mu_R\right)$, and $C_\lambda \equiv C\left(\closure{\mbb{R}},\lambda\right)$ are isometrically isomorphic. Moreover, they are Banach spaces.
\end{lem}
\begin{proof}{:}
Let $f$ be any measurable, bounded, and continuous Borel function with domain $\mc{D}_f$ such that $\mc{D}_f\subseteq C(\closure{\mbb{R}})\to\mbb{R}$. $f$ is continuous if and only if $\lim_{x\to a^-}f(x)=\lim_{x\to a^+}f(x)=f(a)$ for every $a\in \mc{D}_f$. The measures $\mu_L,\mu_R,\lambda$ are mappings from homeomorphic topological spaces $X_{\mu_L},X_{\mu_R},X$ to $\mbb{R}$ respectively. Since the measures $\mu_L,\mu_R\preceq \lambda$ almost everywhere, they are continuous. Thus for any continuous function $f:C\left(\closure{\mbb{R}},\mu\right)\to\mbb{R}$, with $\mu$ either $\mu_L, \mu_R, \lambda$, and we have $\mu_L(f)\to\lambda(f)$ in the strong norm topology, and similarly for $\mu_R$. Thus by continuity $\mu_L(f)=\mu_R(f)=\lambda(f)$, for all $f$ with $\mc{D}_f\in C\left(\closure{\mbb{R}}\right)$. Since $f$ is a linear, norm preserving, continuous mapping such that $f: X_{\mu_L},X_{\mu_R},X\to\mbb{R}$, $f$ must also be injective, which implies that it has a kernel with $ker(f)=\{0\}$. Therefore, for each $x\in Rng(f)$, there exists a $g\in C\left(\mbb{R}\right)$ where $g=f^{-1}(x)$. Thus $f$ is a bijective invertible map over $X_{\mu_L},X_{\mu_R},$ and $X$. Therefore $C_{\mu_L}\cong\ C_\lambda$, and similarly for $C_{\mu_R}$. By transitivity, $C_{\mu_L}\cong C_{\mu_R}$. $C\left(\closure{\mbb{R}}\right)$ inherits the norm from the metric on $C_\lambda$, which is a norm-complete linear metric space. Therefore $C_\lambda$ is a Banach space of bounded continuous functions over $\closure{\mbb{R}}$. It follows that $C_{\mu_L}, C_{\mu_R}$ are also Banach spaces since each is isometrically isomorphic to $C_\lambda$.
\end{proof}
\begin{proof}{:}{(Alt):}
Let be $B\left(C_\lambda,C_{\mu_L}\right)$ denote the space of bounded, linear, continuous mappings from $C_\lambda\to C_{\mu_L}$. Take $f$ to be a bounded, linear, continuous function in $C\left(\closure{\mbb{R}}\right)$. By definition $f$ is Lebesgue integrable, and hence is measurable in $X$ with $\lim_{x\to a^+}f(x)=\lim_{x\to a^-}f(x)=f(a)$. The identity map $e$, acts as an invertible, isometric, bijective linear transformation. Take $e\circ f$ such that $e\circ f:C_\lambda\to C_{\mu_L}$ over the field of scalars $\closure{\mbb{R}}$, implies $f$ is bounded, linear and continuous in $C_{\mu_L}$. Therefore $C_\lambda\cong C\mu_L$. This holds analogously for $C_{\mu_R}$. By transitivity, $C_{\mu_L}\cong C_{\mu_R}$. $C\left(\closure{\mbb{R}}\right)$ inherits the norm from the metric on $C_\lambda$, which is a norm-complete linear metric space. Thus $C_\lambda$ is a Banach space of bounded continuous functions over $\closure{\mbb{R}}$. It follows that $C_{\mu_L}, C_{\mu_R}$ are also Banach spaces since each is isometrically isomorphic to $C_\lambda$.
\end{proof}

\begin{cor}\label{blcf1}
Since $C_\lambda, C_{\mu_L},$ and $C_{\mu_R}$ are isometrically isomorphic Banach spaces, we have trivially have that $\mc{B}\left(C_\lambda ,C_{\mu_L,\mu_R}\right)$ and $\mc{B}\left(C_{\mu_L,\mu_R},C_\lambda\right)$ are also Banach spaces, where $\mc{B}\left(\cdot ,\cdot\right)$ denotes the space of bounded continuous linear mappings.
\end{cor}

\begin{rmk} The Banach spaces $C_{\mu_L,\mu_R}$ are just two copies of the same Banach space $C_\lambda$. This construction is by choice and rather inert for the class of continuous functions on $\closure{\mbb{R}}$. We trivially have $f\in C_\lambda$ \textit{if and only if} $f\in C_{\mu_L}\cap C_{\mu_R}$, by definition. However, this will not be true for more general function(al) spaces.
\end{rmk} 

If we consider the measure spaces $X_{\mu_L,\mu_R}$ and the linear function spaces $C_{\mu_L,\mu_R}$, we may view them as the topological quotient spaces $X/X_{\mu_L,\mu_R}$ and linear function spaces $C_\lambda/C_{\mu_L,\mu_R}$. 

\begin{defin}\label{topq}
We define the equivalence class structures $X_L=X/X_{\mu_L}$ and $X_R=X/X_{\mu_R}$ as the equivalence classes of left and right semi-continuous measure spaces over the measure space $X=\left(\closure{\mbb{R}},\mathscr{B},\lambda\right)$.
\end{defin}   
Similarly,
\begin{defin}\label{linq}
We define the Banach space equivalence class structures $C_L=C_\lambda /C_{\mu_L}$ and $C_R=C_\lambda /C_{\mu_R}$ as the equivalence classes of continuous functions which have equivalent left (resp. right) measures for all $f\in C_\lambda=C\left(X,\lambda\right)\iff f\in C_L\cap C_R$, and $\mu_L (f)=\mu_R (f)=\lambda (f)$.
\end{defin}

\begin{thm}\label{bclf2}
The space of bounded linear continuous mappings $B\left(C_\lambda ,C_{L,R}\right)$, and $B\left(C_{L,R},C_\lambda\right)$ are isometrically isomorphic  Banach spaces with the uniform norm topologies. Moreover, they are locally convex and separable.
\end{thm}
\begin{proof}{:}
The normed linear metric space $C_\lambda$ is homogeneous. Any linear bounded continuous mapping $f:C_\lambda\to C_{L,R}$ can be regarded as a continuous isometric mapping $C_\lambda\to Y\subseteq C_\lambda$. Since any homogeneous space is homeomorphic to itself, the mapping $f$ preserves the norm topologies. Since $f$ is a norm preserving continuous linear map, it is invertible and $f^{-1}$ is continuous on $Y$. Therefore $f$ is bijective. It follows that $B\left(C_\lambda ,C_{L,R}\right)$, and $B\left(C_{L,R},C_\lambda\right)$ are isometrically isomorphic. $C_\lambda$ is complete linear metric space with the uniform norm and therefore is a Banach space. Since $C_\lambda\cong C_{L,R}$, and $C_{L,R}$ are also Banach spaces, which imply $B\left(C_\lambda,C_{L,R}\right)$ and $B\left(C_{L,R},C_\lambda\right)$ are Banach spaces. The properties of convexity and separability can be shown in the standard way, and are omitted in the proof. 
\end{proof}

\subsubsection{$L^p$ function spaces}\label{Lp functions}
We now move on to the spaces of Lebesgue integrable functions, $\left(L^p\right)$, and their duals, the continuous linear functionals acting on $L^p$ functions. This will be slightly more intricate than the spaces of bounded linear continuous functions. The enlarged class of absolutely integrable functions includes discontinuous functions. The Lebesgue and Lebesgue-Stieltjes measures which are built from collections of disjoint intervals of $\closure{\mbb{R}}$ has the effect of changing the measure of discontinuous functions equal almost everywhere for equivalence classes up to LMZ subsets. This matter complicates the construction of equivalence class structures on the space of absolutely integrable functions, and leaves intact the subspaces of continuous functions. 

Let us begin with the usual equivalence class identifications. For two functions $f,g$ in the space of absolutely integrable functions, we make the identification of equivalences classes of $f,g$ such that $\left[f\right]\sim\left[g\right]$ if $f=g$ almost everywhere. Let $L^p\left(X\right)$ denote the quotient space of equivalence classes of the $p$-th power (with $1\leq p\leq\infty$) of absolutely integrable functions over the measure space $X$, with Lebesgue measure $\lambda$. 

Here the function spaces $L^p\left(X,\lambda\right)$ and $L^p\left(X_{L,R}\right)$ are \textbf{\textit{not}} isometrically isomorphic. For example, take two elements from the class of step functions over the intervals $\left[0,1\right]$ and $\left(0,1\right]$. The measure of $H_{L}$ in the latter interval is $\mu_L\left(H_L\right)_{(0,1]}=H_L(1)-H_L(0^+)=1-1=0$, whereas in the former interval $\lambda\left(H_L\right)_{[0,1]}=H_L(1)-H_L(0)=1-0=1$, and therefore $\mu_L(H_L)\neq\lambda(H_L)$.

\begin{defin}\label{LpdX}
In the function space $L^p_{loc}(X)$ with $\left(1\leq p\leq\infty\right)$, we define the subspace of discontinuous functions by $L^p_d\left(X\right)\equiv L^p_{loc}\left(X\right)\setminus\closure{C_\lambda\left(X\right)}$, where $\closure{C_\lambda\left(X\right)}$ denotes the $p$-norm completion of the subspace $C_\lambda\left(X\right)$ over all collections of intervals in $L^p_{loc}\left(X\right)$. \textbf{Note:} $L^p_d\left(X\right)$ contains functions defined by collections of measure zero sets, however they are clearly not dense in $X$. 
\end{defin}

\begin{thm}\label{discLploc}
Any non-LMZ function $f\in L^p_d(X)$ over some bounded interval $I\subset\mbb{R}$ such that $f$ is not left or right semicontinuous on $I$ is \textbf{completely discontinuous}. Furthermore $f$ has measure zero on the subspaces of semicontinuous $L^p_{loc}(X)$-functions in the $p$-norm topology, but not necessarily on $L^p_{loc}$(X) itself. 
\end{thm}
\begin{proof}{:}
Let $L^p_{\mu_L,\mu_R,loc}\left(X_{L,R}\right)$ denote the subspaces of left, respectively right, semicontinuous $L^p_{loc}$ functions over all intervals $I\subseteq\mbb{R}$. We may then define the quotient space equivalence classes $L^p_{loc}(X)/L^p_{\mu_L,\mu_R,loc}(X_{L,R})\equiv L^p_{L,R,loc}$. By def.~\ref{LpdX}, $L^p_d(X)=L^p_{loc}(X)\setminus \closure{C_\lambda(X)}\cong L^p_{loc}(X)\setminus\left(\closure{C_L}\cap \closure{C_R}\right)\cong \left(L^p_{loc}(X)\setminus \closure{C_L}\right)\cup \left(L^p_{loc}\setminus \closure{C_R}\right)$, by Lemma \ref{lemCR}. For any bounded interval $I\subset\closure{\mbb{R}}$, and any semicontinuous function $g$ on $I$, $g$ has the domain $\mc{D}_g=L^p_{L,R,loc}\cup\closure{C_{L,R}}\subset L^p_{loc}$. The equivalence classes defined above are just the $L^p_{loc}$-norm completion of $C_{L,R}$, which are bounded by the uniform $(\sup)$ norms in $C_{L,R}$ and equivalent to the $L^1$-norm in $L^p_{loc}$. Therefore $\mc{D}_g=cl\left(\closure{C_{L,R}}\right)=L^p_{L,R,loc}$, where $cl(\cdot)$ denotes the closure of the subspaces of locally semicontinuous functions over all collections of subintervals $I_i\subset\mbb{\closure{R}}$. Thus, semi-continuity of the linear quotient subspaces $L^p_{L,R,loc}$ implies that $L^p_d\left(X\right)\cap L^p_{L,R, loc}=\emptyset$. Therefore $f$ has measure $\mu(f)=\mu(\emptyset)=0$. 
\end{proof}
For the moment we merely state the following corollary.
\begin{cor}\label{intdiscLpdX}
Let $f,I$ be as in Theorem \ref{discLploc}. If there exists an $I_1\subset I\subset\mbb{R}$ (also for $I_1\supset I$) which allows $f$ to be continuous (or semicontinuous ) over $I_1$, then there is a refinement of the measure gauge $\tilde{\mu}=\mu_{I_1}$ such $\tilde{\mu}(f)\neq 0$. If the refinement is accomplished by adding (or removing) a single point to $I$ (i.e. a set of measure zero), we say that "$f$ is $\mu_L$-\textbf{extendable} (analogously $\mu_R$-\textbf{extendable})" (and respectively "$\mu_L (\mu_R)$-\textbf{restrictable}"). 
\end{cor}  
\begin{proof}{:}
See Theorem \ref{sc-Banach} below.
\end{proof}
\begin{rmk}\label{extconvtn}
Our convention will be to generally apply the term \textbf{extendable} for both of the extension and restriction cases when there is no chance of confusion from the immediate context. If a given set is extendable, we take it as implicit, that this also implies that the set is restrictable, unless otherwise stated.
\end{rmk} 
\begin{cor}\label{reg}
$\mu_{L,R}$-extendable implies regularity.
\end{cor}
\begin{proof}{:}
For any locally measurable function over $X_{L,R}$ the addition or removal of a set of measure zero at an endpoint will extend any half-open measure $\mu_{L,R}$ uniquely such that it is a regular Borel measure measure. Then choose a measure gauge refinement such that $\tilde{\mu}_{L,R}=\mu^\ast$ or $\tilde{\mu}_{L,R}=\mu_\ast$ and $\tilde{\mu}_{L,R}\preceq \lambda$. The uniqueness lies in the designation of the strongest measure topology for which continuity between $X_{L,R}$ and $X$ holds. 
\end{proof}

\begin{prop}{: Properties of $L^p_{\mu_L}\left(X_L\right)$ and $L^p_{\mu_R}\left(X_R\right)$}\label{LpmuLmuR}\\
For $1\leq p\leq\infty$, the linear measure subspaces $L^p_{\mu_L}\left(X_L\right)$ and $L^p_{\mu_R}\left(X_R\right)$ have the following properties.
\begin{enumerate}
\item $L^p_{\mu_L}\left(X_L\right)$ and $L^p_{\mu_R}\left(X_R\right)$ are isomorphic but not isometric.
\item Each are separable seminormed Banach subspaces of $L^p\left(X\right)$.
\item $L^p_{\mu_L}\left(X_L\right)$ and $L^p_{\mu_R}\left(X_R\right)$ are complete linear semi-normed sub-manifolds of $L^p\left(X\right)$, with $L^p_{\mu_L}\left(X_L\right)\perp L^p_{\mu_R}\left(X_R\right)$, in the $p$-norm topology. 
\end{enumerate}
\end{prop}
\begin{proof}{:}
\begin{enumerate}
\item There exists some linear transformation $f:L^p_{\mu_L}\left(X_L\right)\to L^p_{\mu_R}\left(X_R\right)$. Since $f$ is a linear mapping over the same field of scalars in one dimension, they are isomorphic. To see that they are not isometric, see Ex.~\ref{expl3}. 
\item $L^p_{\mu_L}\left(X_L\right)$ and $L^p_{\mu_R}\left(X_R\right)$ are formed by countable disjoint collections of Borel sets, any of which can be taken as a base for $\closure{\mbb{R}}$. Hence they are separable. For any function $f$ over either space, the identity map $e$ will be an isometric bijection of $f$ \textbf{into} a subspace of $L^p(X)$. Since $L^p\left(X\right)$ is a Banach space, so are $L^p_{\mu_L}\left(X_L\right)$ and $L^p_{\mu_R}\left(X_R\right)$. However, $L^p_{\mu_L}\left(X_L\right)$ and $L^p_{\mu_R}\left(X_R\right)$ have only a semi-norm in the $p$-norm topology, as they do not have quotient space identifications to make them a true normed linear subspace of $L^p\left(X\right)$.   
\item Completeness is shown below in Theorem~\ref{sc-Banach}. From 2, they are each semi-normed Banach spaces, and therefore linear. Orthogonality is as follows. There exists linear transformations $f:L^p_{\mu_L}\left(X_L\right)\to L^p_{\mu_R}\left(X_R\right)$ and $g:L^p_{\mu_R}\left(X_R\right)\to L^p_{\mu_L}\left(X_L\right)$ over a base interval $I_b\subset\closure{\mbb{R}}$ which includes a discontinuity, such that $f,g$ are injective, $\mu_{L,I_b}(f)\neq\mu_{R,I_b}(g)\neq 0$, and $f=g$ almost everywhere on $I_b$. Since we have $f=g$ almost everywhere, this implies that there exists an equivalence class over the measure space $X:~[f]\sim[g]$ such that $\lambda_{I_b}([f])=\lambda_{I_b}([g])$ on $L^p\left(X\right)$. By Theorem~\ref{discLploc}, it follows that $f$ is completely discontinuous in $L^p\left(X\right)\setminus L^p_{\mu_L}\left(X_L\right)\cong L^p_{\mu_R}\left(X_R\right)\cup L^p_d\left(X\right)$, and therefore $\mu_{R,I_b}(f)=0$. This holds analogously for $g$ , with $\mu_L(g)=0$. Since $L^p_{\mu_L}\left(X_L\right)$ and $L^p_{\mu_R}\left(X_R\right)$ are Banach subspaces with the norm topology, they are semi-normed spaces. To see that they are sub-manifolds, take the class of step functions which is dense in $L^p\left(X\right)$. For any $\chi_{I_1=[a,b)},\chi_{I_2=[c,d)}\in\chi_R$ over $\closure{\mbb{R}}$, with $a,d<b,c\in\closure{\mbb{R}}$, we have $\chi_{I_1}+\chi_{I_2}=\chi_{I_1\cup I_2}\in\chi_R$. For any scalar $a\in\closure{\mbb{R}}$, then it follows that $a\chi_{I_1}\in\chi_R$. This holds for $\chi_L$ as well.    
\end{enumerate}
\end{proof}
For any two measurable functions $f,g$ in $L^p_{\mu_L}\left(X_L\right)$ or $L^p_{\mu_R}\left(X_R\right)$, we cannot form the equivalence class identifications $[f]\sim[g]$ with $f=g$ almost everywhere. We have seen that it is possible to have $f=g$ a.e., but $\mu(f)\neq \mu(g)$ or even $\mu(f)\neq 0$ but $\mu(g)=0$. We must be more discriminatory in defining the quotient space equivalence class identifications.

\begin{defin}{: $L^p_L\left(X_L\right)$ and $L^p_R\left(X_R\right)$ quotient spaces}\label{LpLRQclss}\\
We define the quotient space spaces of $L^p_L\equiv L^p\left(X\right)/L^p_{\mu_L}\left(X_L\right)$ and $L^p_R\equiv L^p\left(X\right)/L^p_{\mu_R}\left(X_R\right)$ by identifying any two functions, say $f,g\in L^p_L$ where $f=g$ almost everywhere \textbf{and} if either $f$ \textbf{or} $g$ is $\mu_L$-extendable such that $\mu_L(f)=\mu_L(g)\in L^p_L$. We denote the left semicontinuous equivalence identifications by $[f]\sim_L [g]$. We analogously define the right semicontinuous equivalence classes and denote the right semicontinuous equivalence identifications by $[f]\sim_R[g]$. 
\end{defin}

\begin{thm}\label{sc-Banach}
 Each semicontinuous measure space is a complete dense subspace of semicontinuous functions of $L^p_{loc}\left(X\right)$, which is also semicontinuous with respect to the uniform norm topologies on the measure spaces $X_{L,R}$. The $p$-norm is the completion of the semicontinuous subspaces with respect to the uniform metric norms on the measure spaces $X_{L,R}$. We denote the semicontinuous subspaces of $L_{loc}^p\left(X\right)$ by $L^p_{L,loc}$ and $L^p_{R,loc}$ respectively. $L_{L,loc}^p$ and $L_{R,loc}^p$ are quotient space equivalence classes which extend uniquely over all of $L^p\left(X\right)$, such that $L^p_{L,loc}\equiv L^p_{loc}(X)/L^p_{\mu_L,loc}(X_L)$ and $L^p_{R,loc}\equiv L^p_{loc}(X)/L^p_{\mu_R,loc}(X_R)$, where $L^p_{\mu_L,\mu_R,loc}(X_{L,R})$ are the $\mu_L,\mu_R$-measure subspaces of $L^p_{loc}(X)$. Furthermore, under these extensions, all $L^p$-functions are either piecewise semicontinuous or completely (left and right) continuous.
\end{thm}

\begin{proof}{:}\ref{sc-Banach}
The density is straightforward. Take the class of step functions $\chi_I$, over some interval $I\subseteq\closure{\mbb{R}}$. We tacitly assume the interval $I$ to be congruent with the measure space topology over which $\chi_I$ is associated within the immediate adjacent text. It is well known that $\chi_I$ is dense in $L^p$ and there for dense in all $L^p$ subsets. Take $\{\chi_I\}\in L^p_{loc}\left(X\right)$ for $I=\cup_{\mc{I}_i\subseteq \closure{\mbb{R}}}^\infty \mc{I}_i$, where the unions are taken over all disjoint collections of Borel generating sets in $X$. Thus for any $f\in L^p_{loc}\left(X_I\right)$, with $I$ a bounded subset, there is a Cauchy sequence of step functions over $I$ such that $\{\chi_{I_n}\}\to f$ a.e. as $n\to\infty$. Take the gauge for the Lebesgue measure to be the strongest topology for which given any $\eps>0$, there is a finite sub-cover of $I=\cup_{\mc{I}_i\subseteq\closure{\mbb{R}}}^m I_i$, congruent with the metric topology on $X$, such that for any $x\in X$ and $n>m$, then $\left|\chi_{I_n}(x)-f(x)\right|\leq\eps$ as $n\to\infty$, and the sub-cover $I$ contains the limit point of $\{\chi_{I_n}\}=\accentset{\circ}{f}(x)$. For any $p$-th power, the sequence $\{\chi_{I_n}\}\to f$ pointwise monotonically for each sub-interval. By linearity, the sequence is continuous (and also sequentially compact). Continuity (of a sequentially compact set) and the sub-additivity of the measure, imply uniform convergence and therefore the series of the subsequences $\Sigma_{n_i}^\infty \left|\chi_{n_i}\right|^p$ is term-wise $p$-summable. It follows that $\{\chi_{I_n}\}\to f$ uniformly in the $p$-norm topology, and therefore it is complete in $L^p(X)$. Similar to definitions \ref{topq} and \ref{linq}, the spaces of $L^p_{\mu_L,\mu_R, loc}\left(X_{L,R}\right)$ mappings are absolutely integrable over the measure spaces $X/X_{\mu_L,\mu_R}$, with $\mu_L,\mu_R\preceq \lambda$ for all choice of gauges such that the mapping $f:L^p_{loc}(X)\to L^p_{L,R,loc}(X_{L,R})$ also is topologically continuous. Therefore the continuity for any pre-image $f^{-1}:L^p_{\mu_L,\mu_R,loc}\left(X_{L,R}\right)\to L^p_{loc}\left(X\right)$ is continuous in the stronger Lebesgue-Stieltjes measure topologies, $\mu_L,\mu_R$, as well as $\lambda$. By Theorem \ref{discLploc}, the semi-continuity of the linear mappings $L^p_{\mu_L,\mu_R,loc}\left(X_{L,R}\right)$ imply that $L^p_d\left(X\right)\cap L^p_{\mu_L,\mu_R, loc}\left(X_{L,R}\right)=\emptyset$ \textbf{however}, they may share intervals with common interiors. Let $f\in L^p_d\left(X\right)$ be $\mu_{L,R}$-extendable. We denote the subset of all $\mu_{L,R}$-extendable functions as $L^p_{d, ext}$. For any Cauchy $\{\chi_I\}$ and therefore also any $f\in L^p_{d, ext}$, $f$ can be uniquely  extended by of a set of measure zero such that $\closure{L^p_{d,ext}}=L^p_{L,R,loc}$. Thus, any extension by a collection of sets of measure zero, will be a unique extension from $L^p_d$ to the left or right semicontinuous function spaces, such that for any completely discontinuous function $f$ over the interval $(a,b)$, then $\closure{f_{(a,b)}}_L=f_{(a,b)\cup[b]}\in L^p_{loc}\left(X_L\right)$ and similarly,  $\closure{f_{(a,b)}}_R=f_{(a,b)\cup[a]}\in L^p_{loc}\left(X_R\right)$. Since it is generally true that $L^p\left(Y\right)\supset L^p_{loc}\left(Y\right)$ for some arbitrary measure space $Y$, then the quotient space of semicontinuous functions $L^p_{L, loc}$ can be extended to all of $L^p\left(X\right)$ up to an equivalence relation almost everywhere. The process can be repeated for any piecewise defined function on $L^p(X)$ defined over all collections of bounded subintervals of $\closure{\mbb{R}}$. In this way, all subintervals $I_i\subset \mbb{R}$ can be reduced to arbitrarily small but countable lengths. Therefore the subspace $\closure{L^p_{d}(X)}$ may be reduced to functions taking values from collections of nowhere dense sets. Since we regard $\mu_L=\mu_R=\lambda$ on sets of measure zero, every $L^p$-function may be regarded as either semicontinuous or completely (left and right) continuous almost everywhere.
\end{proof}

\begin{prop}{: Properties of $L^p_{L,R}$}\label{LpLR-props}\\
For $1\leq p\leq\infty$, the quotient subspaces $L^p_{L}(X_L)$ and $L^p_{R}(X_R)$ have the following properties in $L^p(X)$.
\begin{enumerate}
\item $L^p_L(X_L)$ and $L^p_R(X_R)$ are isomorphic, but \textbf{not} isometric. 
\item Each are separable normed Banach sub-spaces of $L^p\left(X\right)$.
\item $L^p_L(X_L)$ and $L^p_R(X_R)$ are complete sub-manifolds of $L^p\left(X\right)$ with $L^p_L(X_L)\perp L^p_R(X_R)$.
\end{enumerate}
\end{prop}
\begin{proof}{:}
Analogous to Proposition~\ref{LpmuLmuR}. After making the identifications of equivalence classes defined by the quotient spaces in Def.~\ref{LpLRQclss}, we have bona fide norms on $L^p_L\left(X_L\right)$ and $L^p_R\left(X_R\right)$. 
\end{proof}

At this point we have constructed a nice formalism of sub-manifolds within the classical $L^p\left(X\right)$ space of functions. It is interesting to point out that the Banach spaces $L^p_{L,R}$ include the Banach spaces of $C_{L,R}\left(X_{L,R}\right)$. This can be seen by viewing the spaces $L^p_{L,R}\left(X_{L,R}\right)$ as just the $p$-norm completions of $C_{L,R}\left(X_{L,R}\right)$ in $L^p\left(X\right)$. However, we seem to have much more. The $\mu_{L,R}$-extendable functions are also mappings into subspaces $C_{L,R}\left(X_{L,R}\right)$ as well. Thus we have a partial embedding from subspaces of $L^p$ to subspaces of $C_\lambda \left(X\right)$. 

\begin{thm}{: $C_{L,R}\left(X_{L,R}\right)\hookleftarrow L^p_{L,R}\left(X_{L,R}\right)$}\label{embedLpC}\\
Let $f$ be a non-atomic, completely discontinuous, and piecewise defined function on a subspace of either $L^p_L\left(X_L\right)$ or $L^p_R\left(X_R\right)$. If $f$ is $\mu_{L,R}$-extendable, then the left (right) extended mappings $f_L,f_R$ are \textbf{partial embeddings} of equivalence classes from $L^p_{d,ext}\left(X_L\right)$ into subspaces $C_L\left(X_L\right)$ (or collections of left semicontinuous intervals of $C_L\left(X_L\right)$) and $L^p_{d,ext}\left(X_R\right)$ into subspaces of $C_R\left(X_R\right)$ (or collections of right semicontinuous intervals of $C_R\left(X_R\right)$) with the relative uniform $\|\cdot\|_{\sup}$ topology on $C_L,C_R$ respectively. Moreover, $C\left(\closure{\mbb{R}}\right)\subsetneq C_{L,R}\left(X_{L,R}\right)$ in general.
\end{thm}
\begin{proof}{:}{ (Left case)}\\
Let $f$ be non-atomic, completely discontinuous, and piecewise in $L^p_{d,ext}$. Then $f$ is $\mu_L$-extendable such that $f\to f_{L}\in L^p_L\left(X_L\right)$, for example, a collection of left semicontinuous step functions. $f_L$ need only be non-zero on some collection of left semicontinuous intervals $I_L=\cup_{i}(\cdot,\cdot]_i\subset \closure{\mbb{R}}$, provided that on each subinterval $(\cdot,\cdot]_i$ where $f\neq 0$, $f$ is continuous and, $f=0$ on the remaining left semicontinuous intervals. Since functions which admit left semicontinuity are regarded as continuous in $C_L\left(X_L\right)$, restricting to the $L^1_L\left(X_L\right)$ norm is equivalent to the uniform $\|\cdot\|_{\sup}$ norm on $C_L\left(X_L\right)$. In the H\"{o}lder extremal case, left (semi)continuity then implies that $\|f\|_{\infty}\leq \|f\|_1$ is $L^1$-norm bounded wherever $f\neq 0$. Hence $f_L$ is Lebesgue-Stieltjes integrable, which implies it is Riemann-Stieltjes integrable, and norm bounded by the relative $\|\cdot\|_{\sup}$ in $C_L\left(X_L\right)$. Therefore we have a complete normed space. The right semicontinuous case is analogous. By Theorem~\ref{sc-Banach}, up to functions defined on sets of measure zero, we have a topologically continuous partial embedding of functions of $L^p(X_L)\hookrightarrow C_{L,R}\left(X_{L,R}\right)$. In particular for $p<\infty$, this shows that the classes of semicontinuous step functions $\chi_{L,R}$, which are dense in $L^p\left(X\right)$, are similarly dense in $C_{L,R}\left(X_{L,R}\right)$, though $\chi_{L,R}\notin C\left(\closure{\mbb{R}}\right)$ in general. This implies that $C\left(\closure{\mbb{R}}\right)\subsetneq C_{L,R}\left(X_{L,R}\right)$ in general.  
\end{proof}
We mention that as a result of the imposed topological and measure space congruence, Theorem~\ref{embedLpC} implies that the embedding holds both algebraically and topologically. This will have interesting implications regarding the dual spaces of $L^p_{L,R}$. For example, it is well known that the dual of $L^\infty$ is not generally $L^1$, but does contain $L^1$ as a subspace\cite{roman2,Teschl14}. The dual of $L^\infty$ may be characterized such that $L^{\infty\ast}\cong L^1\cup \mc{BV}$, where $\mc{BV}$ is the space of functions of bounded variation with $\|\cdot\|_{\mc{BV}}$-norm. $\mc{BV}$ includes the subset of finitely additive signed and countably additive measures. Let $B\left(\Sigma\right)$ be the set of bounded $\Sigma$-measurable functions, and $\mu$, a measure on the space $B\left(\Sigma,\mu\right)$. The measurable space, $B\left(\Sigma\right)^\ast \cong\mc{BV}\left(\Sigma\right)$ is the continuous dual of $B\left(\Sigma\right)$. It follows that $B\left(\Sigma,\mu\right)^\ast\cong\mc{BV}\left(\Sigma,\mu\right)$ holds with respect to the $\|\cdot\|_{\sup}$-norm, if and only if $B\left(\Sigma, \mu\right)$ is continuous in the $\sup$-norm topology. With the \textit{essential sup}-norm $L^\infty(\mu)\cong B\left(\Sigma\right)/N_\mu$, where $N_\mu$ is the closed subspace of all bounded null measurable functions\cite{Teschl14}. It follows that $L^{\infty\ast}\cong N^\perp_\mu$, the orthogonal complement of $N_\mu$, which is the space of all finitely additive measures on $\Sigma$ that are absolutely continuous in measure ($\mu$-a.c.). If the measure space $B\left(\Sigma, \mu\right)$ is $\sigma$-finite (and therefore separable), then we can identify $L^1(\mu)^\ast\cong L^\infty(\mu)$. Taking the dual once more, we have that $L^1(\mu)\subset L^1(\mu)^{\ast\ast}\cong L^\infty(\mu)^\ast$ by the Radon-Nikodym theorem. Given Theorem~\ref{embedLpC} and the class of step functions: $\chi\in \mc{BV}$, this suggests that there is also some non-trivial embedding of $\mc{BV}\left(X_{L,R}\right)\hookrightarrow C_{L,R}\left(X_{L,R}\right)^\ast$, provided any finitely additive measure $\mu\in\mc{BV}$ is $\mu_{L,R}$-extendable.

\subsection{Linear transformations on the semicontinuous Banach spaces}\label{Lplinearxforms}

In~\ref{Lp functions} we established that $L^p_{L,R}\left(X_{L,R}\right)$ are sub-manifolds of $L^p\left(X\right)$, however we would also like to define the  transformations acting on them. These are Banach spaces (subspaces and subalgebras) of continuous linear transformations and operators, which we denote as $\mc{B}$ and $\mc{B}_{L,R}$. 
\begin{defin}{: Banach spaces of linear transformations}\\
The space of linear transformations from $L^p_{L,R}\left(X_{L,R}\right)\to L_{L,R}^p\left(X_{L,R}'\right)$ is also a Banach space with the $p$-norm topology. We denote the Banach space of left (resp. right) linear transformations by $\mc{B}_L\left(L_{L}^p\left(X_{L}\right),L_{L}^p\left(X_{L}'\right)\right)$ and $\mc{B}_R\left(L_{R}^p\left(X_{R}\right),L_{R}^p\left(X_{R}'\right)\right)$. \textbf{Note:} Reflection symmetry (parity) is not preserved, and is the result of an odd number of parity violating transformations (i.e. transformations composed of an odd number of reflections). We distinctly denote parity violating linear transformations by $\mc{B}_{L\to R}\left(L_{L}^p\left(X_{L}\right),L_{R}^p\left(X_{R}'\right)\right)$, or $\mc{B}_{R\to L}\left(L_{R}^p\left(X_{R}\right),L_{L}^p\left(X_{L}'\right)\right)$, respectively.
\end{defin}
\begin{rmk}
In general the space of linear transformations $B\left(X,Y\right)$ which maps the arbitrary measure space $X\to Y$ is a Banach space if and only if $Y$ is a Banach space, where given some $f\in B$, $f\left(Y\right)$ inherits the relative topology from $Y$\cite{roman2}. Since $L^p_{L,R}\left(X_{L,R}\right)$ are Banach spaces, we take this as a definition for $\mc{B}$, and $\mc{B}_{L,R}$.
\end{rmk}
\begin{defin}{: Reflection (parity) Map}\label{parity}\\
Let $f_L,f_R$ be semicontinuous over some interval $I\subset\closure{\mbb{R}}$, with $f_L\in\mc{B}_L(I_{(\cdot,\cdot]}))$ and $f_R\in \mc{B}_R(I_{[\cdot,\cdot)})$, and $\hat{\Pi}$ be the parity operator define on elements $x\in X,X_{L,R}$ such that $\hat{\Pi}: x\mapsto -x$. Therefore $\hat{\Pi}:X_{L,R}\to X_{R,L}$, where $L,R$ designates that $\hat{\Pi}:X_L\to X_R$ and $R,L$ designates $\hat{\Pi}:X_R\to X_L$. For functions $f\in\mc{B},\mc{B}_{L,R}$, $\hat{\Pi}$ acts through composition $(\hat{\Pi} f)(x)\equiv(f\circ \hat{\Pi})(x)=f(-x)$. If $f_L,f_R$ are not invariant under reflections, then $\hat{\Pi} f_L$ defines a mapping $\mc{B}_{L\to R}\left(L^p_L(I_{(\cdot,\cdot]}),L^p_R(I_{[\cdot,\cdot)})\right)$ and $\hat{\Pi} f_R$ defines a mapping $\mc{B}_{R\to L}\left(L^p_R (I_{[\cdot,\cdot)}), L^p_L (I_{(\cdot,\cdot]})\right)$. 
\end{defin} 
\begin{cor}{: Measures of $\hat{\Pi}$}\label{Opi}\\
Let $f_L,f_R$ be as in def.~\ref{parity}, in particular not invariant with respect to $\hat{\Pi}$. It follows that $\mu_L \left(\pi f_L\right)(I_{(\cdot,\cdot]})=\mu_R\left(\pi f_R\right)(I_{[\cdot,\cdot)})=0$.
\end{cor}
\begin{proof}{:} A consequence of $\hat{\Pi} f_L\in L^p_d\left(X,\mu_L\right)$ and $\hat{\Pi}f_R\in L^p_d\left(X,\mu_R\right)$, when $f_L,f_R$ are not invariant under reflections of the domain coordinate.
\end{proof}
We see that $\hat{\Pi}$ can act as a linear operator on functions of $\mc{B}_{L,R}$, provided it acts on functions which are invariant under reflections. Otherwise, $\hat{\Pi}$ acts as a linear transformation, mapping subsets of $\mc{B}_{L,R}$ to subsets $\mc{B}_{R,L}$.  

We also have that $\mc{B}_L$ and $\mc{B}_R$ are left and right semicontinuous Banach subalgebras. Let $f,g$ be families of left and right semicontinuous functions and $a\in\mbb{R}$. Then $f_1+f_2:\mc{B}_L\times \mc{B}_L\to\mc{B}_L$ and $af:\mc{B}_L\times \mbb{R}\to\mc{B}_L$ are well defined and closed under addition and scalar multiplication. The same holds for the family of right semicontinuous functions, $g$. Moreover, if we take $f_1,f_2\in L^p_L\left(X_L\right)$, it follows that may define pointwise multiplication over their common domains ($\mc{D}_{f_1\cdot f_2}=\mc{D}_{f_1}\cap\mc{D}_{f_2}\neq\{\emptyset\}$), as $f_1$ and $f_2$ are continuous in $L^p_L\left(X_L\right)$ and therefore, $f_1+f_2:L^p_L\left(X_L\right)\times L^p_L\left(X_L\right)\to L^p_L\left(X_L\right)$. We may take this one step further to define $f_1+g_1:L^p_L\left(X_L\right)\times L^p_R\left(X_R\right)\to L^p_L\left(X_L\right)\oplus L^p_R\left(X_R\right)$, as well as $f_1g_1:L^p_L\left(X_L\right)\times L^p_R\left(X_R\right)\to L^p\left(X\right)$. This is really nothing new, as the sum and products of two absolutely integrable functions is also integrable\cite{roman2,Teschl14}. However subspaces of $L^p_{L,R}\left(X_{L,R}\right)$ having embeddings in $C\left(\closure{\mbb{R}}\right)$ can have non-trivial implications for the continuous dual space of linear functionals, to which we now turn.

\subsection{The duals of $L^p\left(X\right)$ and $L^p_{L,R}\left(X_{L,R}\right)$}\label{secLp*} 
We now wish to consider the dual space of our Banach spaces. These spaces are the spaces of continuous linear functionals on $L^p\left(X\right)$ and the semicontinuous submanifolds $L_{L,R}^p\left(X_{L,R}\right)$. Let $p,q\in\mbb{N}:1<p,q<\infty$. From the H\"{o}lder inequality, it is well known that $L^p\left(\closure{\mbb{R}}\right)^\ast \cong L^q\left(\closure{\mbb{R}}\right)$ where $p,q$ are conjugate pairs such that $\frac{1}{p}+\frac{1}{q}=1$. Since our universal space is $L^1\left(X\right)$ where $X=\left(\closure{\mbb{R}}, \mathscr{B},\lambda\right)$, the well known duals apply for $1\leq p\leq\infty$ remain true. Therefore in addition to dual spaces stated for $1<p<\infty$, and since the Lebesgue measure $\lambda$ is $\sigma$-finite, we know that $L^1\left(X\right)^\ast\cong L^\infty\left(X\right)$. It follows that $L^1\left(X\right)^{\ast\ast}\cong L^\infty\left(X\right)^\ast\supset L^1\left(X\right)\cup \mc{BV}\left(X\right)$, where we have identified $\mc{BV}$ as the set which contains all finitely additive signed Borel measures and all countable finitely additive signed Borel measures which are absolutely continuous ($\mu$-a.c.) as a subspace\cite{Teschl14}. See the discussion after the proof of Theorem~\ref{embedLpC}. Our measure space is generated by the collection of all Lebesgue measures over $\closure{\mbb{R}}$, which is equivalent to the measure space generated by all collections of all Borel measures. Therefore we can make the identification that $B\left(\mathscr{B},\lambda\right)=L^1\left(X\right)^\ast$, where $B\left(\mathscr{B},\lambda\right)$ is the Banach space of all Lebesgue measurable functions\footnote{Here we note that $\mc{B}\left(\mathscr{B},\lambda\right)=\closure{\mc{B}\left(\mathscr{B},\mu\right)}$, is the Borel $\mu$-norm completion of the measure space. But $\|\cdot\|_{L^1}=\|\cdot\|_{\sup}$, which is the metric norm for $\mu$.}, but with the total variation metric norm $\left|\nu\right|\left(Y\right)\equiv\sup\{\sum_{i=1}^n \left|\nu\left(Y_i\right)\right|\big|Y_i\in \mathscr{B}~\text{disjoint},~Y=\cup_{i=1}^n Y_i\}$. Note that the definition of the $\left|\nu\right|\left(Y\right)$ does not require $\sigma$-additivity, so $\nu$ is finite if $\left|\nu\right|\left(Y\right)<\infty$. $\nu$ is therefore a map $\nu:\mathscr{B}\to \mc{C}\left(\mc{R}\right)$, and is referred to as a \textit{(complex) content}. It follows that $\left|\nu\right|$ is a \textit{positive content}\cite{Teschl14}. We do not need to say anything further regarding $L^p\left(X\right)^\ast$ for $1\leq p\leq\infty$, and now turn our attention towards finding the duals for $L^p_{L,R}\left(X_{L,R}\right)$.

\subsubsection{$L^p_{L,R}\left(X_{L,R}\right)^\ast$, with $1\leq p<\infty$.}
Given that we are working with $L^p$ submanifolds defined on semicontinuous metric-norm quotient spaces of $L^p\left(\closure{\mbb{R}}\right)$, these dual spaces require slightly more care in their definitions. First we make the following definition.
\begin{defin}{: Continuous linear functionals on $L^p\left(X\right),~p<\infty$}\label{Lp*}\\
For the Banach space $L^p\left(X)\equiv L^p(X,\lambda\right)$ with the $\sigma$-finite Lebesgue measure $\lambda$, $p<\infty$, and $q$, the dual conjugate to $p$ such that $\frac{1}{p}+\frac{1}{q}=1$, then we define the space of continuous linear functionals $L^q\left(X\right)$ to be the collection of all maps $g\in L^q\left(X\right)\mapsto l_g\in L^p\left(X\right)^\ast$ given by 
\beq\label{linf}
l_g\left(f\right)&=\int_X gfd\mu(g)
\eeq
\end{defin}

We state the well known theorem for the Banach spaces $L^p(X)$:

\begin{thm}{: $L^q\left(X\right)$ properties}\label{Lqspace}\\
Let $L^p\left(X\right), l_g, p,q$ be as in def.~\ref{Lp*}. For $1\leq p<\infty$, then the collection of all mappings defined by Eq.~\eqr{linf} is an isometric isomorphism, and thus $L^p\left(X\right)^\ast\cong L^q\left(X\right)$, and reflexive for $1<p<\infty$. If $p=\infty$, then the mapping $l_g$ is isometric, but not isomorphic.
\end{thm}
\begin{proof}{:} See \cite[theorem 11.1 and corollary 11.2]{Teschl14}. 
\end{proof}
\begin{rmk} The isometric isomorphism for $p=\infty$ is discussed in~\ref{secLp*}
\end{rmk} 
With def.~\ref{Lp*} and Theorem~\ref{Lqspace}, we may now define the dual of $L^p_{L,R}\left(X_{L,R}\right)$ for $1\leq p<\infty$.
\begin{thm}{: The dual space of $L^p_{L,R}\left(X_{L,R}\right)$, with $1\leq p<\infty$}\label{LqLR}\\
Let $p,q$ be the H\"{o}lder conjugate pairs, with $1\leq p<\infty$. The continuous isometric isomorphic dual to the semicontinuous Banach spaces $L^p_{L,R}\left(X_{L,R}\right)$ is given $L^q_{L,R}\left(X_{L,R}\right)$, for $p=\infty$, they are only isometric. Again, for $1<p<\infty$, the spaces $L^p_{L,R}\left(X_{L,R}\right)$ are reflexive.
\end{thm}
\begin{proof}{:}{($\left(L^p_L\left(X_L\right)\right)^\ast$)}
We recall that the Banach spaces are equivalence classes define by the topologies of quotient space $L^p_L\left(X_L\right)\equiv L^p\left(X/X_L,\mu_L\right)\sim L^p\left(X,\lambda\right)/L^p\left(X_L,\mu_L\right)$ where the metric norm is inherited from the relative topology of $X_L\equiv\left(\closure{\mbb{R}},\mathscr{B},\mu_L\right)$ and $\mu_L\preceq \lambda$. $L^p_R\left(X_R\right)$ was analogously defined. They are also closed subspaces (submanifolds) of $L^p\left(X\right)$, see discussion following Corollary~\ref{Opi}. Since for an arbitrary Banach space $Y$ with closed subspace $M_Y$, the dual of the quotient space $\left(Y/M_Y\right)$ is $\left(Y/M_Y\right)^\ast\cong\{l\in Y^\ast\big| M_Y\subseteq \ker(l)\}$, the analogous result must hold here. Thus we have that $L^p_L\left(X_L\right)^\ast\cong \{l\in L^p\left(X\right)^\ast \big| L^p \left(X_L,\mu_L\right)\subseteq\ker(l)\}$. But $\left(L^p\left(X\right)/L^p\left(X_L,\mu_L\right)\right)^\ast\cong \left(L^p_L\left(X_L\right)\right)^\perp$, which is the annihilator set of $L^p_L\left(X_L\right)$. Therefore $\left(L^p_L\left(X_L\right)\right)^\perp=\{l\in X^\ast\big| l(x)=0~\forall x\in L^p_L\left(X_L\right)\}$. But this is just the set of functionals $l$, which takes $g\in L^q\left(X\right)\mapsto l_g\in L^p\left(X\right)^\ast$, for which any $f\in L^p_L\left(X_L\right)$ is null for the functional $l_g(f)=\int_X gfd(\cdot)=0$. Since $L^p_L\left(X_L\right)\perp L^p_R\left(X_R\right)$ by Propositions~\ref{LpmuLmuR},~\ref{LpLR-props}, this implies that the annihilator set is just the union of the set of mappings with measure $\mu_R$ and the zero functional $[0]$. Thus by the H\"{o}lder inequality, $\left(L^p_L\left(X_L\right)\right)^\ast\cong L^q_L\left(X_L\right)$, with $\frac{1}{p}+\frac{1}{q}=1$. By Theorem~\ref{Lqspace}  $L^p_L\left(X_L\right)\cong L^q_L\left(X_L\right)$: they are isometrically isomorphic, and for $1<p<\infty$, they are reflexive. For $p=\infty$, they are isometric. The proof of $\left(L^p_R\left(X_R\right)\right)^\ast$ follows analogously.   
\end{proof}

\subsubsection{The dual space of $L^\infty_{L,R}\left(X_{L,R}\right)$}\label{LinftyLR}
Now we discuss the dual space of $L^\infty_{L,R}\left(X_{L,R}\right)$. In short, this will be similar to what is to be expected from $L^\infty\left(X\right)^\ast$, modulo minor modifications regarding the quotient space constructions of $L^\infty_{L,R}\left(X_{L,R}\right)$. To reassure ourselves that all the relevant details are taken into consideration, we will proceed constructively. It is rather harmless to assume that $L^\infty_{L,R}\left(X_{L,R}\right)^\ast \subset L^\infty\left(X\right)^\ast \cong \mc{BV}\left(X\right)$, consistent with standard results from analysis. This leads us to the following.
\begin{thm}{:}\label{muLR121BV}
Let $\mu_{L,R}$ left and right semicontinuous Borel measures, and $f_L,f_R$ be left and right continuous functions respectively in $\mc{BV}\left(\closure{\mbb{R}}\right)$. There is a one-to-one correspondence between functions $f_L\in\mc{BV}\left(\closure{\mbb{R}}\right)$, and $f_R\in\mc{BV}\left(\closure{\mbb{R}}\right)$ which are left, respectively right continuous, and normalized by $f_L(0)=0$, $f_R(0)=0$, and complex Borel measures $\mu_L$ and respectively $\mu_R$ on $\mbb{R}$ such that $f_L$ is the left continuous distribution function of $\mu_L$ defined by 
\beq\label{Ldistfunc}
f_L(x)\xleftrightarrow{1-to-1} \mu_L(x)\equiv\begin{cases}
-\mu_L\left((x,0]\right),	&  x<0,\\
0,					& x=0,\\
\mu_L\left((0,x]\right),	& x>0,
\end{cases}
\eeq
and similarly, $f_R$ is the right continuous distribution function of $\mu_R$ defined by
\beq\label{Rdistfunc}
f_R(x)\xleftrightarrow{1-to-1} \mu_R(x)\equiv\begin{cases}
-\mu_R\left([x,0)\right),	& x<0,\\
0,					& x=0,\\
\mu_R\left([0,x)\right),	& x>0,\\
\end{cases}
\eeq
It follows that the distribution functions of the total variations of $\mu_L,\mu_R$ are respectively defined by
\beq\label{Ltotvar}
\left|\mu_L\right|(a)=\lim_{x\to  a^-}V_{(0,x]}(f_L)=V_{(0,a]}(f_L),
\eeq
and
\beq\label{Rtotvar}
\left|\mu_R\right|(a)=\lim_{x\to 0^+}V_{[x,a)}(f_R)=V_{[0,a)}(f_R).
\eeq
\end{thm}
\begin{proof}{:}{~$\mu_R,~~(\rightarrow)$}\\
Each right continuous complex measure $df_R$ can be identified with a function $f_R\in\mc{BV}$. Assume $f_R$ is normalized. Then by construction $f_R$ is equal to the right continuous distribution function.\\
 $(\leftarrow)$\\
 Let $d\mu_R$ be a complex measure with distribution function $\mu_R$. For each $a<b\in\mbb{R}$, which has the interval partition $P=\{a=x_0,\ldots,x_n=b\}$. It follows that the total variation, $V_{[a,b)}(\mu_R)=\sup_{P}V\left(P,\mu_R\right)=\sup_P\sum^n_i\left|\mu_R\left([x_{i-1},x_i)\right)\right|\leq \left|\mu_R\right|\left([a,b)\right)$, which is of bounded variation. This can also be extended to all Borel sets. First consider a measure $\mu(x)$ with total variation $V_{[0,x)}(\mu_R)$. Now $\mu$ is inner regular with respect to $\mu_R$, and thus valid for all open subsets of a compact interval $I\subset\mbb{R}$. Extend this to all Borel sets by outer regularity. It then follows that $\mu=\left|\mu_R\right|$, which implies that $\left|\mu_R\right|(x)=V_{[0,x]}(f)$. The case for left continuous Borel measures follows analogously.
 \end{proof}

So we have for any Borel measure, a unique left, and a unique right continuous function in $\mc{BV}$. As before, we form quotient spaces for the left and right continuous measure spaces $X_{L,R}$, such that functions that are equal almost everywhere in measure, with respect to the left and right Borel measures are identified. 

\begin{defin}\label{BVLR}
For the measure spaces $X_{L,R}$, we denote the left and right semicontinuous sets of $BV$ over $X_{L,R}$ respectively by defining the quotient spaces $\mc{BV}_L\equiv\mc{BV}_L\left(X_L\right)\sim\mc{B}\left(X,\lambda\right)/\mc{BV}\left(X_L,\mu_L\right)$ and $\mc{BV}_R\equiv\mc{BV}_R\left(X_R\right)\sim\mc{BV}\left(X,\lambda\right)/\mc{BV}\left(X_R,\mu_R\right)$. These quotient spaces identify functions which are almost everywhere equivalent and continuous with respect to $\mu_L$ and $\mu_R$ respectively, such that for $f_L\in\mc{BV}_L$ and $f_R\in\mc{BV}_R$, $\mu_L(f_L),\mu_R(f_R)\neq 0$ and $\mu_L(f_R)=\mu_R(f_L)=0$.
\end{defin}

Let $I=[a,b]\subset\mbb{R}$ be a bounded interval. It is a well known result from analysis that the set $\left(\mc{BV}[I],\|f\|_{\mc{BV}}\right)$ is a Banach space, with norm defined $\|f\|_{\mc{BV}}\equiv \left|f(a)\right|+V_{I}(f)$. Now that we have $\mc{BV}_{L,R}$ defined, we may see that they are also Banach spaces bounded above by the $\|f\|_{BV}$-norm. We will return to this momentarily. For now let us exploit the freedom granted us by continuity of our Banach spaces. 

We recall that any Borel measure $\mu$ is absolutely continuous ($\mu$-a.c.) with respect to Lebesgue measure $\lambda$, if and only if its distribution function is \textit{locally absolutely continuous} ( i.e. absolutely continuous on every compact sub-interval). The consequence of this is and the Radon-Nikodym derivative, is that $\mu$ is differentiable a.e., such that
\beq\label{Borelmuint}
\mu(x)=\mu(0)+\int_0^x\mu'(y)dy,
\eeq
$\mu'$ integrable, and $\int_{\mbb{R}}\left|\mu'(y)\right|dy=\left|\mu\right|(\mbb{R})$. However this is just the fundamental theorem of calculus, which provides an alternative definition $\mu$-a.c. functions. Since we have a one-to-one correspondence between half-open Borel measures and semicontinuous functions of $\mc{BV}$, we may then characterize the half-open Borel measures in terms of some unique primitive function associated with the integral of Eq.~\eqr{Borelmuint}.
\begin{thm}{:}\label{muAC}
On the quotient spaces of $\mc{BV}_{L,R}\left(X_{L,R}\right)$, any semicontinuous function is absolutely continuous with respect to Lebesgue measure. 
\end{thm}
\begin{proof}{:} Recall that $\mu_{L,R}\preceq \lambda$ by construction. Then each Borel measure $\mu_{L,R}$ is uniquely associated with some primitive left $(f_L)$ or right $(f_R)$ continuous function.
\end{proof}

Theorem~\ref{muAC}, and the preceding discussion gives us everything that we need to complete the discussion for continuous dual of $L^\infty_{L,R}\left(X_{L,R}\right)$.
\begin{thm}{: $L^\infty_{L,R}\left(X_{L,R}\right)^\ast$}\label{Linfin*}\\
$L^\infty_{L,R}\left(X_{L,R}\right)^\ast\cong\mc{BV}_{L,R}\left(X_{L,R},\|\cdot\|_{\mc{BV}}\right)$, where $\|\cdot\|_{\mc{BV}}$ is the norm completion of $\mc{BV}_{L,R}\left(X_{L,R}\right)$. Moreover, the bi-dual of $L^\infty_{L,R}\left(X_{L,R}\right)$ is precisely the set $L^1_{L,R,loc}\left(X_{L,R}\right)\hookrightarrow C_{L,R,c}\left(X_{L,R}\right)$, where the embedding is continuous and dense.
\end{thm}
\begin{proof}{:}
Here we implicitly assume that we are on the measure spaces $X_L$ or $X_R$, and omit their explicit mention in the Banach spaces. We start with $\left(L^1_{L,R}\right)^\ast\cong L^\infty_{L,R}$. Therefore we have the inclusions
\beq\label{1}
L^1_{L,R}\subset \left(L^1_{L,R}\right)^\ast\cong L^\infty_{L,R}
\eeq
Taking the dual, we have from Theorem~\ref{muAC} that $\left(L^\infty_{L,R}\right)^\ast\cong \mc{BV}_{L,R}$. The dual of this gives
\beq\label{2}
\mc{BV}_{L,R}\subset\left(\mc{BV}_{L,R}\right)^\ast\cong\left(L^\infty_{L,R}\right)^{\ast\ast}\subset \left(L^\infty_{L,R}\right)^\ast\cong\mc{BV}_{L,R}.
\eeq
Therefore $\mc{BV}_{L,R}$ is self-dual. Eq.~\eqr{1} also implies
\beq\begin{aligned}\label{3}
L^1_{L,R}\subset\left(L^1_{L,R}\right)^{\ast\ast}\subset\left(L^1_{L,R}\right)^\ast\cong L^\infty_{L,R}&\implies L^1_{L,R}\subset\left(L^1_{L,R}\right)^{\ast\ast}\subset\left(L^\infty_{L,R}\right)^\ast\cong\mc{BV}_{L,R}\\
&\implies L^1_{L,R}\subset \left(L^1_{L,R}\right)^{\ast\ast\ast}\cong\left(L^\infty_{L,R}\right)^{\ast\ast}\subset \left(\mc{BV}_{L,R}\right)^\ast\cong\mc{BV}_{L,R}\\
&\implies L^1_{L,R}\subset \left(\mc{BV}_{L,R}\right)^\ast\subset \left(L^\infty_{L,R}\right)^{\ast\ast}\cong\left(\mc{BV}_{L,R}\right)^\ast\cong\mc{BV}_{L,R}\\
&\implies L^1_{L,R}\subset\left(L^1_{L,R}\right)^{\ast\ast}\cong\mc{BV}_{L,R}\cong\left(\mc{BV}_{L,R}\right)^\ast
\end{aligned}\eeq
The last line above allows us to identify the locally Lebesgue measurable functions with the $\|\cdot\|_{\sup}$-norm as a subset of the dual to $\mc{BV}_{L,R}$ functions, which are the continuous functions over all compact intervals of $\closure{\mbb{R}}$, denoted as $C_{L,R,c}\left(X_{L,R}\right)$. We can see this by noting that for all inclusions above, each set inclusion is dense with respect to the corresponding superset. Next, the $L^1_{L,R,loc}$ functions can be $\mu_{L,R}$-extended by Theorem~\ref{embedLpC}. Since $L^1_{L,R,loc}$ has $\|\cdot\|_{\sup}$-norm, and 
\beq\label{L1BVnorm}
\|\cdot\|_{p}\leq\|\cdot\|_{\sup}\leq \|\cdot\|_{\mc{BV}}=\|\cdot\|+\|V_{L,R}\|_{\sup},
\eeq 
where $\|V_{L,R}\|_{\sup}$ denotes the supremum of the left/right variation. Hence for $1\leq p\leq\infty$, we have $\|\cdot\|_{L^p_{L,R,loc}}$ is bounded by $\|\cdot\|_{\mc{BV}}$. It follows that the $\mc{BV}$-norm is the norm completion for $\left(L^p_{L,R,loc}\right)^\ast$ and therefore, for all $L^p_{L,R}$. Thus $L^p_{L,R,loc}\hookrightarrow C_{L,R,c}$ continuously.  
\end{proof}
\begin{rmk}\label{mixnorm} After the initial posting of this work to the arXiv, the author was made aware of the work of Johnson and Lapidus, by a form student of M. Lapidus. Particularly, the norm above is very similar in form to the \textit{mixed-norm} defined in \cite[Ch. 15.2]{johnlap}. However there are some differences, which differ mainly in their respective origins based on how the linear function spaces are fundamentally structured. We will not discuss these details further here.  
\end{rmk}

\section{The Dirac-$\delta'$ System}\label{deltaprm}
We will now utilize the formalism developed in~\ref{semi} to analyze the quantum system described by Eq.~\eqr{hamdelta2}, which we reproduce below. The system under investigation here is given by the quantum mechanical Hamiltonian in 1-dimension described by Schr\"{o}dinger's equation. In what follows, we will only discuss the so called \textit{"interaction Hamiltonian"}, where the potential is assumed to contribute to the functional equation: $\hat{V}(x)\neq 0$. 

In order to make the equation well defined, we consider as a linear functional given by the first integral equation derived from Eq.~\eqr{varH}. If the system is Hamiltonian (at least locally), there exists a vector flow whose first integral is the solution to Hamilton's equations. Thus, integration is implicit in the construction of the functional equation. 
\beq\label{schrodeq}
\hat{H}&=\frac{1}{2m}\hat{P}^2+V(\hat{x})\\
&=-\frac{\hbar^2}{2m}\frac{\partial^2}{\partial x^2}+\alpha\frac{\partial\delta (x)}{\partial x}.
\eeq
$\hat{P}$ is the one dimensional momentum operator, $i\hbar\frac{\partial}{\partial x}$. $\alpha$ is a coupling constant with unspecified sign ($\left|{\alpha}\right| >0$) and units of $length^2\times energy$. 

\subsection{Differential Geometry and the Hamiltonian Operator}\label{diffgeo}
For the moment we recall some generalities of Hamiltonian systems on differentiable manifolds in order to redefine Eq.~\eqr{schrodeq} in terms of operators acting on them. Let $M$ be the compactified real line in $n$-dimensions. The Hamiltonian functional $\hat{H}$ defines a map $\hat{H}:TM\to M$, and thus $\hat{H}\in T^\ast M\cong TM$. We identify the generalized coordinates $(q_1,\ldots,q_n)$ on $M$ as the configuration space of the manifold, such that for $x\in M$, then $x=x(q_i)$. Then $\hat{H}(M)$ is a linear functional on the tangent bundle $TM$ of M. Let $\mc{X}(M)$ denote the space of a vector fields in $TM$, and $F^p(M)$ denote the space of $p$-forms on $M$.  

The solutions to Eq.~\eqr{schrodeq} are defined on the space of $p$-forms $\mb{F(M)}=\bigoplus^n_{p=0}F^p(M)$, in terms of the scalar product. For two forms $\phi,\psi\in F^p(M)$, the scalar product is 
\beq\label{eigenH}
\braket{\hat{H}\psi,\phi}=\int_{M}\hat{H}\psi\wedge\ast\phi=E\psi[\ast\phi],
\eeq
where $E$ is the eigenvalue of $\hat{H}$, and $\ast$ the Hodge dual. In the case of vacuum to vacuum transitions, then $\psi=\phi$, for a vacuum state $\psi$ and $\braket{\psi,\psi}\geq 0$ for all $\psi$. In terms of scalar solutions, then $\psi\in F^0(M)$.

The dimensionless free kinetic energy operator $\hat{H}_f$, is equivalent to the Laplace-Beltrami (Laplacian) operator, which defines a map from $F^p(M)\to F^p(M)$, for $0\leq p\leq n$. For a $C^\infty(M)$ scalar function $f\in F^0(M)$,
\beq\notag
\mathbold{d}f(q): T_{q_i} M\to T_{f(q_i)}M,~\quad 1\leq i\leq n
\eeq
such that at the point $q_i$, $\mathbold{d}f(q_i)\in T^\ast_{q_i} M$, the cotangent bundle, and $T_{f(q_i)}M = T_{q_i}(TM)$ is the tangent space of $TM$ at $q_i$. Any free vacuum to vacuum scalar wave function which satisfies Eq~\eqr{schrodeq} with $\hat{V}(x)=0$ is restricted to the class of harmonic 0-forms, $\mb{H}^0$. This must also be the case for $\hat{V}(x)\neq 0$, otherwise by the Hodge decomposition theorem, any $\psi\in F^p(M)$ with $p>0$ will necessarily be orthogonal to $\mb{H}^0$, resulting in a decoupled ($i.e.$ non-interacting) solution set\footnote{We exclude cases such as tensor products of $\mbb{R}\times\ldots\times\mbb{R}$}. 

In order for the Hamiltonian to admit an interactive scalar solutions and avoid the introduction of the 1-form "potential", $dx\frac{\partial\delta (x)}{\partial x}$, we take Eq.~\eqr{schrodeq} to be defined as
\beq\label{schrodiffop}
\hat{H}&:=\frac{\hbar^2}{2m}\mathbold{\Delta}+\alpha (\ast \mathbold{d}\delta(x)),
\eeq
On $\closure{\mbb{R}}^1$, $\ast\mathbold{d}\delta(x)=\ast \left(\frac{\partial}{\partial x}\delta(x)\right)=\frac{\partial\delta(x)}{\partial x}$ is a 0-form rather than the component of a 1-form. It follows that $\hat{H}$ defines a map, $\hat{H}: F^p\to F^p$, from which Eq.~\eqr{schrodeq} follows directly. Therefore, we take Eq~\eqr{schrodeq} to implicitly have the intent defined by Eq.~\eqr{schrodiffop}. 

For the remainder of this subsection, it will be convenient to set all the constants above to 1, and discuss Eq.~\eqr{schrodiffop} with respect to a general potential $\tta\in F^0(M)$ rather than specifically having $\hat{V}(x)=\ast \frac{\partial \delta(x)}{\partial x}$. We will also restrict our discussion to $M=\closure{\mbb{R}}^1$, so $P(M)=M^2$. Then we have $\hat{H}$ given by
\beq\label{schrodiffop2}
\hat{H}\to\mathbold{\Delta}+\ast(\mathbold{d}\tta).
\eeq

In order to discuss possible harmonic solutions to Eq.~\eqr{schrodiffop2}, we first need to define some differential equivalence class relations. Equivalence class identifications may seem somewhat unnecessary. However, because the function equivalence classes established in \secref{Lp functions} exclude identifications on sets of LMZ, and such identifications must be established under alternative associations. This is important for limiting processes, such as derivatives, convergence of regularized sequences of nets to distributions, and sheafs. In these cases, there is some form of $\epsilon$ neighborhood for which we wish to include the limit point, $\epsilon =0$. Without such equivalence class identifications it is not necessarily true, that a sequence of approximating functions which ordinarily converge to a distribution at a measure zero limit point, may be identified with a distribution.

For example, take $f(x)=\delta(x)$ and $g_\epsilon (x)=\frac{1}{2}\tan^{-1}({\frac{x}{\epsilon}})$, and $\frac{1}{2}\int_{-\infty}^\infty dk~ e^{ikx}\cdot e^{\epsilon\left|k\right|}=\frac{\epsilon}{x^2+\epsilon^2}$. In this case the regulated Dirac-$\delta$, given by $\delta_\epsilon (x)=\frac{\epsilon}{x^2 +\epsilon^2}$, may be integrated to produce $g_\epsilon (x)$. Since $g_\epsilon (x)\to \sgn(x)$ as $\epsilon\to 0$, and thus $\delta_\epsilon^{-1}(x)=g_\epsilon (x)$. However, one would like to have $f^{-1}(x)=\delta^{-1}(x)=\frac{1}{2}\sgn(x)$ at the limit point of $\epsilon=0$, as in~Ex. \ref{expl3}. The equivalence class identifications permit such connections to be established at $\epsilon=0$, without explicitly mapping LMZ sets under the function(al) equivalence classes of \secref{Lp functions}.

Suppose that in some open neighborhood $x\in U\subset M$ (with a given topological space $M$), we have $\mathbold{d}f(x)=g(x)\in M$ which defines a class of differentiable equivalences $:\mathbold{d}f\sim_x g$. We can repeat this process for any $x\in U^\prime\subset U$. Furthermore, on the spaces of vector fields and differential forms, the equivalence relation $\sim_x$ defines the stalk $\mc{F}_x:=\mathbold{d}f_x$ of the presheaf on open neighborhoods of $U$. For $M$ a differentiable manifold, this is the space of jets of order $k$, $J^k_x(M,M)$. This is particularly true for distributions on the tangent spaces (vector fields) of open subsets of $U\subset T_x M$ and codistributions on open subsets of $V\subset T(T_x M)=T^\ast_x M$. If it so happens that $f\in\mb{H}^0(\closure{\mbb{R}}^1)$, then $g\in \mb{H}^1(\closure{\mbb{R}}^1)$. It follows that $\mathbold{d}\ast g=\mathbold{d}\mathbold{\delta}f\in \mb{H}^0(\closure{\mbb{R}}^1)$ which is coexact and coclosed, and establishes another differential equivalence relation. Thus, we have a second order equivalence relation for some $h\sim_x \mathbold{\delta}g\sim_x \mathbold{\delta}\mathbold{d}f$, on all open subsets of $V$. If a sheaf is established, then we have a form of uniqueness given by the equivalence relation. We now give the formal definition.

\begin{defin}\label{germequiv}{:} Let $\alpha$ be a 0-form for which $\mathbold{d}\alpha=\tta$ is exact. In the category of differential forms, the equivalence class of differentiable maps $f,g\in U\subset M$, for an open neighborhood $U$ of a differentiable manifold $M$ at the point $x$, is the set of  germs, given by $\mathfrak{f}'_x:=[f']_x\sim_x [g]$. Denote the equivalence class $[f']$ on the space of differentiable forms by $f\to\alpha$ and $~'\to \mathbold{d}$, then this equivalence class defines a germ with primitive $\tta:~[\mathbold{d}\alpha]_x\sim_x[\tta]$ in the stalk $\mc{F}'_x$ of the presheaf $\mc{F}', ~\forall U^\prime\subset U\in M$. Similarly for vector fields, $\mathfrak{X}'_x$ denotes the stalk of the equivalence class of vector fields such that, for $X,Y\in\mathfrak{X}'$, then $[X']_x\sim_x [Y]$ for the presheaf $\mathfrak{X}$. If the equivalence relations hold at each $x\in U^\prime,~\forall U^\prime \subset U$, then $\mc{F}'$ is a sheaf. Then $\mathfrak{X}^{',p}$ and $\mc{F}^{',p}$ denotes the space of $p$-dimensional vector fields and the dual space of $p$-dimensional forms respectively.
\end{defin} 

Recently a diffeomorphism invariant full sheaf property was established in \cite[Def.17-Prop.19]{2016arXiv161106061N} utilizing Colombeau algebras. There, a similar set of identifications to those made in Def.~\ref{germequiv} are established generally for regularizible generalized functions. 

With the previous definition, we now consider $\tta\in\mc{F}'_x$ in Eq.~\eqr{schrodiffop2} and look for local harmonic forms defined $[\tta]_x\in \mathfrak{f}'_x\subset F^1(M)$.

\begin{prop}{:} Let $\omega\in\mathbold{H}^0(\closure{\mbb{R}}^1)$, the space of harmonic 0-forms on $\closure{\mbb{R}}^1$, and bounded on $\closure{\mbb{R}}^1$. The Hamiltonian (Eq.~\eqr{schrodiffop2}) admits non-trivial local solutions (on $\closure{\mbb{R}}^1)$, if and only if the 0-form $\tta$ of the potential (given by $\hat{V}(x)=\ast(\mathbold{d}\tta)$ is exact and closed with respect to some exact 1-form $[\mathbold{d}\alpha]_x$, such that $\ast[\mathbold{d}\alpha]_x =[\tta]_x$, and $[\tta]_x$ is coexact and coclosed. By non-trivial, we mean that $\tta\neq 0$ and $\tta\in\mathbold{H}^0$.\end{prop} 
\begin{proof}{:}{$(\Rightarrow)$}
The forward direction is trivial. Assume that $[\tta]\in\mb{H}^0$. Since $\omega$ is harmonic, then the scalar product $\braket{\hat{H}\omega, \omega}$ becomes $\braket{\ast(\mathbold{d}\tta)\omega, \omega}$. If $[\tta]\in\mb{H}^0$, then we necessarily have $\mathbold{d}[\tta]=\mathbold{\delta}[\tta]=0$. Therefore $\mathbold{d}[\tta]$ is closed and coclosed. This implies that there exists some 0-form $\alpha$, such that $\mathbold{d}\alpha\in F^1$ and $\ast\mathbold{d}\alpha=[\tta]\in F^0$. So $[\tta]$ is exact and trivially coexact, since for all $\beta\in F^0$, $\delta\beta=0$.

$(\Leftarrow)$
Let $[\tta]\sim_x\ast[\mathbold{d}\alpha]$, with $\mathbold{d}\alpha$ an exact 1-form. Then the inner product becomes
\beq\begin{aligned}
\braket{(\ast\mathbold{d}[\tta])\omega,\omega}&=\braket{(\ast\mathbold{d}\ast[\mathbold{d}]\alpha)\omega,\omega}\\
&=\braket{(\mathbold{\delta}[\mathbold{d}\alpha])\omega,\omega}\\
&=\braket{(\mathbold{\Delta}\alpha)\omega,\omega}
\end{aligned}\eeq
We may rewrite ($\mathbold{\Delta}\alpha)\omega$ in the last line above as
\beq\begin{aligned}
(\mathbold{\Delta}\alpha)\omega &=\mathbold{\Delta}(\alpha\omega)-\alpha\mathbold{\Delta}\omega-2\mathbold{d}\alpha\mathbold{\delta}\omega-2\mathbold{\delta}\alpha\mathbold{d}\omega.
\end{aligned}\eeq
The last two terms in the previous line vanish trivially by the fact that $\alpha,\omega\in F^0$, so $\mathbold{\delta}\alpha=0$ and similarly for the term with $\mathbold{\delta}\omega$. Further more, the terms are necessarily orthogonal to $\mb{H}^0$ by the Hodge decomposition theorem. Therefore, the last line above reduces to,
\beq\begin{aligned}\label{Lapalpha}
 \braket{(\mathbold{\Delta}\alpha)\omega,\omega}&=\braket{\mathbold{\Delta}(\alpha\omega),\omega}-\braket{\alpha\mathbold{\Delta}\omega,\omega}\\
 &=\braket{\alpha\omega,\mathbold{\Delta}\omega}-\braket{\alpha\mathbold{\Delta}\omega,\omega}.
\end{aligned}\eeq
The final line above follows because $\mathbold{\Delta}$ is self adjoint. Since $\omega$ is harmonic both terms above are in $\mb{H}^0$, and therefore we have a solution in $\mb{H}^0$ for the left hand side. 
\end{proof}

\begin{rmk} Let us comment on the above proposition. First, the proposition does not imply that $\tta\in\mb{H}^0(M)$, but rather that $\tta$ is equivalent to a derivative of an exact form $\alpha$, which is harmonic. Thus the above proposition first maps to $\tta$ to its primitive function (its $\mathfrak{f}'_x$ germ equivalence class), which is harmonic. Second, we note that at no point do require anything regarding the smoothness of $\alpha$, only that it be differentiable. The proof does not depend upon being able to use integration by parts. Therefore this includes distributional derivatives for which an equivalence in $\mathfrak{f}'$ can be established. Although we did not directly invert the linear function $\tta$, the semicontinuous Banach spaces here do allow for this option. However, we do need to be in a linear space where this option exists in order to justify the equivalence classes above.   
\end{rmk}
 
\begin{cor}{:}\label{hamdiffopequiv}
Let $\hat{H}$ be given as in Eq.~\eqr{schrodiffop2}, $\tta,\alpha\in F^0(\closure{\mbb{R}}^1)$ such that $[\tta]_x\sim_x\ast[\mathbold{d}\alpha]$ as in Prop.~\eqr{germequiv}, and either $\omega\in\mb{H}^0\cap C_0^\infty(\closure{\mbb{R}}^1)$ or $\braket{\alpha\omega,\omega}<\infty$. Then $\hat{H}$ reduces to 
\beq\begin{aligned}\label{redH}
\hat{H}&=\mathbold{\Delta}-\frac{1}{2}\alpha\mathbold{\Delta}\\
&=(1-\frac{1}{2}\alpha)\mathbold{\Delta}.
\end{aligned}\eeq
\end{cor}
\begin{proof}{:}
This is almost a trivial consequence of Prop.~\eqr{germequiv} and $\omega\in\mathbold{H^0}\cap C^\infty_0(\closure{\mbb{R}}^1)$. In Eq.~\eqr{Lapalpha}, $\braket{\mathbold{\Delta}(\alpha\omega),\omega}$ becomes a boundary term in the inner product by
\beq\begin{aligned}
\braket{\mathbold{\Delta}(\alpha\omega),\omega}&=\frac{1}{2}\braket{\mathbold{\Delta}(\alpha\omega),\omega}+\frac{1}{2}\braket{\alpha\omega,\mathbold{\Delta}\omega}\\
&=\frac{1}{2}\mathbold{\Delta}\left(\braket{\alpha\omega,\omega}\right).
\end{aligned}\eeq
Taking $\mathbold{\Delta}$ inside the integral produces the boundary term, which must vanish because $\omega$ vanishes. Alternatively, $\braket{\cdot,\cdot}\in\mbb{R}$, then $\mathbold{\Delta}\braket{\cdot,\cdot}=0$ trivially. Therefore, we may drop the condition that $\omega\in C^\infty_0(\closure{\mbb{R}}^1)$ provided that $\braket{\alpha\omega,\omega}<\infty$. The rest obviously follows. The factor of 1/2 in Eq.~\eqr{redH} arises from moving to the local generalized coordinates in the adjoint map of the canonical cotangent projection, $i.e.$ Hamilton's equations on $T^\ast_{q,p}(T^\ast (\closure{\mbb{R}}^1))$. We will discuss Hamilton's equations in the next section.
\end{proof}

\subsection{Hamilton's Equations on the Cotangent Bundle}\label{Hamcotan}
Recall the phase space of the Hamiltonian system $P(M)$ is a $2n$-dimensional symplectic manifold, such $(q_i,p^i),~1\leq i\leq n$ is the phase space coordinates are identified through the preimage of canonical projection on the cotangent bundle, assuming that $(dq^i,dp_i)$ is the local basis for $T^\ast_{q_i,p^i}M$. Thus, $\pi^{-1}(T\ast_{q_i}M)=T^\ast_{(q_i,p^i)}M$ such that for $\pi (T\ast_{(q_i,p^i)}M)=(q_i,0)$. Let $\beta$ be a $p$-form on $T^\ast_q M$, then $\pi^\ast \beta$ is the pull-back which defines the $p$-form on $T^\ast_{q,p}M$. This is just the standard fibration of $M$ over the cotangent bundle, which is naturally endowed with the fundamental symplectic 2-form structure $\Omega=\sum_i dp_i\wedge dq^i$. 
 
The Hamiltonian $\hat{H}$, is a 0-form on the cotangent bundle. Thus $\mathbold{d}\hat{H}$ is identified as the Poincare\'{e} 1-form on the cotangent bundle. In the generalized local coordinates $(q_i,p^i)$, the Hamiltonian in natural units given by Eq.~\eqr{schrodeq} is
\beq\begin{aligned}\label{localH}
 H_{(q,p)}=\frac{1}{2}p^2+V(q).
\end{aligned}\eeq
Let us define two Hamiltonians $H_1, H_2$ as
\begin{align}\label{locH1}
H_1&=\frac{1}{2}p^2+\frac{\partial\tta(q)}{\partial q},\\
\shortintertext{and,}
\label{locH2}
H_2&=\left(1-\alpha(q)\right)\frac{1}{2}p^2,
\end{align}
where we identify $[\tta]_q\sim_q \ast[\mathbold{d}\alpha]$ as in \secref{diffgeo}. This implies $\ast\mathbold{d}\tta\to\frac{\partial\tta}{\partial q}$ in Eq.~\eqr{locH1}. Eq~\eqr{locH2} is obtained through the defined equivalence class and the results found in Corr.~\ref{hamdiffopequiv}, then $\ast[\mathbold{d}\alpha]\to\frac{1}{\sqrt{2}}\alpha p$ in the local coordinates. Note that this last statement is the origin of the factor of 1/2 which appears in Corr.~\ref{hamdiffopequiv}. Our goal here is to show that Eqs.~\eqr{locH1} and \eqr{locH2} produce equivalent sets of Hamilton's equations, which we will now show.

\begin{prop}{:}\label{Hequiv} The Hamiltonian given by \eqr{locH2} defines a symplectomorphism of Eq.~\eqr{locH1}, which is a first integral along the flow generated by the vector field $X=-\Omega^{-1}\mathbold{d}H_2$, where $\Omega=dp\wedge dq$ is the symplectic 2-form on the phase space $P(\closure{\mbb{R}}^1)=T^\ast_{q,p}\closure{\mbb{R}}^2$.
\end{prop}
\begin{proof}{:}
We need to show that a map $\phi:P\to P: H_1\to H_2$ is canonical. A canonical transformation preserves the Poisson brackets, and therefore is a symplectomorphism on $P$. If $\mathbold{d}H_2$ is closed along the vector field $X$, then $H_2$ is a first integral of $X$.

We begin by finding the differentials associated to each Hamiltonian, given by $\mathbold{d}H=\frac{\partial H}{\partial q}dq+\frac{\partial H}{\partial p}dp$, and the corresponding equations of motion. Thus
\begin{align}
\mathbold{d}H_1=\frac{\partial^2\tta}{\partial q^2}dq+pdp,
\end{align}
which produces the equations of motion
\beq\begin{aligned}\label{locH1eqom}
\frac{dq}{dt}&=\frac{\partial H_1}{\partial p},~\quad   &\frac{dp}{dt}&=-\frac{\partial H_1}{\partial q}\\
&=p		&~		\quad					&=-\frac{\partial^2}{\partial q^2}\cdot\tta\\
\end{aligned}\eeq
Similarly for $H_2$, we have
\begin{align}
\mathbold{d}H_2=-\frac{1}{2}\frac{\partial\alpha}{\partial q}p^2dq+(1-\alpha)pdp,
\end{align}
which yields the equations of motion
\beq\begin{aligned}\label{locH2eqom}
\frac{dq}{dt}&=(1-\alpha)p,~~~\quad 		&\frac{dp}{dt}&=-\frac{1}{2}\frac{\partial \alpha}{\partial q}p^2
\end{aligned}\eeq
Then $\phi:P\to P: q=q',~p\to p'=p:\frac{\partial}{\partial q}[\tta(q)]=\frac{1}{\sqrt{2}}p[\tta(q)]\to\frac{1}{2}\alpha(q)p^2$, is the corresponding map $\phi:H_1\to H_2$. 

It is true locally that the difference between the Poincar\'{e} 1-forms of $H_1$ and $H_2$ is a canonical transformation if the corresponding difference is exact. Let $\omega dt$ be the difference between the Poincar\'{e} 1-forms obtained from $H_1$ and $H_2$ respectively. Thus we have the total time differential as,
\beq\begin{aligned}\label{Poincare1s}
\omega dt&=p\frac{dq}{dt}dt-p'\frac{dq'}{dt}dt\\
	&=\frac{\partial H_1}{\partial p}\frac{dq}{dt}-\frac{\partial H_2}{\partial p}\frac{dq'}{dt}\\
	&=p^2dt-(1-\alpha)^2 p'^2dt\\
	&=p^2 dt-p^2 dt\\
	&=0.
\end{aligned}\eeq
Therefore the transformation is canonical, and also closed. The above is equivalent to the Poisson brackets $\left\{H_1,H_2\right\}_{P.B.}=0$. Since the Poisson brackets are equal to zero, this implies that $H_2$ is a constant (i.e a first integral) along some locally Hamiltonian vector flow, $X$.

We determine $X$ from,
\beq\begin{aligned}\label{Xfield}
X&=-\Omega^{-1}\mathbold{d}H_2\\
  &=-(1-\alpha)p\frac{\partial}{\partial q}+\frac{1}{2}\frac{\partial \alpha}{\partial q}p^2\frac{\partial}{\partial p}.
\end{aligned}\eeq
It follows that 
\beq\begin{aligned}\label{ddH2}
\mathbold{d}i_X\Omega&=\mathbold{dd}H_2\\
					&=\mathbold{d}\left((1-\alpha)pdp-\frac{1}{2}\frac{\partial\alpha}{\partial q}p^2dq\right)\\
					&=-\left(\frac{\partial \alpha }{\partial q}\right)pdq\wedge dp-\left(\frac{\partial\alpha}{\partial q}\right)pdp\wedge dq\\
					&=-\left(\frac{\partial \alpha }{\partial q}\right)pdq\wedge dp+\left(\frac{\partial\alpha}{\partial q}\right)pdq\wedge dp\\
					&=0
\end{aligned}\eeq
Therefore the 1-form $\mathbold{d}H_2$ is closed. This is equivalent to the Lie derivative $L_X\mathbold{d}H_2=0$, which implies energy conservation. Therefore $H_2$ is a first integral along the locally Hamiltonian vector field $X=-\Omega^{-1}\mathbold{d}H_2$.
\end{proof}

Let us discuss the results of the previous two sections in a bit more detail. Clearly these results only apply locally, or ultra-locally. An important implicit assumption is that the equivalence classes exist and admit identifications of $[\tta]_q\sim_q \ast\mathbold{d}\alpha$, which acts as identification of a functional with the derivative of its primitive functional. This is precarious especially with respect to singular distributions. By construction, the results have been derived from the pull-back of some mapping $\phi^\ast$ to the cotangent bundle, which we can always make well defined locally. The equivalence class simply defines an identification between principal fiber in tangent space with the canonical projection of its lift to the cohomology class representatives in the cotangent bundle, fibrated over each point in the base\cite{faddeev199540}.  

Ideally, we would like to push-forward to the tangent bundle, or the base space by $\phi_\ast \hat{H}=\hat{H}\circ\phi^{-1}$. $\hat{H}$ is a 0-form by definition. Therefore we must have $\phi^{-1}$ exist. The mapping $\phi$ implicitly assumes that we have invertible transformations $[\tta]\xleftrightarrow[\phi^{-1}]{\phi} \ast[\mathbold{d}\alpha]$. The implicit assumption of the existence of the inverse restricts this mapping to (sub)spaces on which they are defined. However in the last few sections, we spoke generally of $\tta$, where the potential was defined by $\hat{V}(q)=\ast\mathbold{d}\tta$. Thus, if $[\tta]$ is a globally defined smooth function without singularities, then $\phi$ is a globally defined and invertible map. The map $\phi$ given by Prop.~\ref{Hequiv} is a fiber homomorphism on the cotangent bundle, by $\phi^\ast:T^\ast_{q,p}(\mc{L}^\ast(\closure{\mbb{R}}^2),\Omega_{(\tta(q),p)})\to T^\ast(\mc{L}^\ast(\closure{\mbb{R}}^2)),\Omega_{([\tta(q)],p)}\Large|_U)$, for some $U\subset T^\ast\mc{L}(\closure{\mbb{R}}^1)$, the space of linear functionals. In particular, $\phi$ establishes a covariant connection on the space of jets as in \cite{Vershik1988}. Moreover, the fiber homomorphism maintains the unique point $q\in \closure{\mbb{R}}^1$ identification in the base space. In this sense, $\phi$ is involutive.

However, if $\hat{V}(q)=\delta'(q)$ as in Eq.~\eqr{hamdelta2}, then this is not so. We must restrict $\phi$ to spaces where we can establish the equivalence relation of the Dirac delta with the derivative of its primitive. We saw that this is possible locally and uniquely in \secref{semi}. \textit{De Rham's} theorem applies locally, and results in a Pfaffian solution on a foliated submanifold of the phase space. In particular, $p$-forms are the spaces of linear functionals, which form a module over the cotangent bundle. The only derivations that map $[F^p]\to [F^p]$ on a finite dimensional, $C^{k+1}$ manifold is zero, for $1\leq k<\infty$(Corr. 4.2.39,\cite{von1981differential}). Therefore, the defined equivalence class is non-trivial only for $C^\infty$ maps over manifolds, which is precisely the space of distributions. Therefore, the established equivalence could only make sense if it relates distributions. It then follows that we have $\phi: T^\ast_{q,p}(\mc{L}^\ast(C^\infty_2),\Omega_{(\tta(q),p)})\to T^\ast_{q,p}(\mc{L}^\ast(C^\infty_{2,L,R}),\Omega_{([\tta(q)],p)})$. 

Finally we remark that as a consequence of preserving the Poisson brackets, the map $\phi$ given in Prop.~\ref{Hequiv} defines a "Lie algebra" homomorphism on the phase space $(P(\closure{\mbb{R}}^1),\Omega)$, such that for $\omega\in \mb{H}^0$, then
$\mb{H}^0(P)\to$~diff$(P,\Omega);~~\omega\to X_\omega,~~\mathbold{d}\omega = i_{X_\omega}\Omega$, with kernel the constant functions on each maximal connected component. This essentially makes a claim regarding an "algebra" over the space of functionals (distributions), which is generally difficult to define consistently. At the moment, we do not speculate on the algebraic implications of the above. As the mapping ($\phi$) could be seen as an attempt to define an indefinite integral for distributions (though we regard the mapping as a nuanced, but distinct process), and leave those investigations for future work.

\subsection{The Hamiltonian functional equation}\label{hamfunc}
Let the ket state be an unspecified wave function represented by $\ket{\psi}$. We assume \textit{a priori}, that it is defined over a compatible domain, which remains to be determined. The configuration space (position) $x$, is continuously parametrized by an independent time parameter $t$, ensuring that the energy is a constant of motion with respect to time ($\frac{dE}{dt}=0$). Therefore, we implicitly define the wave function in Dirac's notation, as the position $x$ at time $t$, such that $\ket{\psi_t}\sim\ket{x_t}$. We begin with infinitesimal time shifts of the wave function in the Heisenberg picture. The state $\psi'$ at time $t+\delta t$ (an infinitesimal time shift) is obtained from the state $\psi$ at time $t$ by the perturbative expansion~\cite{ramondfield}
\begin{align}\label{wftdt}
\ket{\psi_{t+\delta t}}\approx \ket{\psi_t}+\frac{i}{\hbar}\delta t \hat{H}\ket{\psi_t}+\mathcal{O}(\delta t)^2.
\end{align}
This implies that the transition amplitude is given by
\begin{align}\label{amp1}
\braket{\psi'_{t+\delta t}| \psi_t}&=\braket{\psi'_t|\psi_t}-\frac{i}{\hbar}\bra{\psi'_{t}}\hat{H}\ket{\psi_t}\delta t+\mathcal{O}(\delta t)^2,
\end{align}
where $\bra{\psi'_t}=\bra{\psi(x')_t}$.

The configuration states of the system at a time $t$, must obey the relations
\begin{align}\label{psi}
\psi&\sim\ket{\psi_t}\\\label{psi*}
\psi^\ast&\sim\bra{\psi_t}\\\label{xop}
\hat{x}\ket{\psi_t}&=x\ket{\psi_t}\\\label{deltanorm}
\braket{\psi_t | \psi'_t}&=\delta(x-x')\\\label{compset}
\int^\infty_{-\infty}\ket{\psi_t}\bra{\psi_t}&=1.
\end{align}
Eqs.~\eqr{psi} and~\eqr{psi*}, are identifications of the particle-state correspondence, Eq.~\eqr{deltanorm} defines the orthonormal Fourier basis with the Dirac-$\delta$ normalization, and Eq.~\eqr{compset} is the completeness relation.

In what follows, we will drop the explicit $\mathcal{O}(\delta t)^2$ term, and tacitly assume it remains present. In terms of a continuous linear functional, the bra-kets must contain information about the measure (space), and must belong to some linear vector space. Therefore Eq.~\eqr{amp1} has the interpretation as a continuous linear functional with measure $\mu\left(\braket{\psi_{t+\delta t}|\psi_t}\right)=1$, of the form
\beq\label{func-amp1}
\mu\left(\braket{\psi'_{t+\delta t}|\psi_t}\right)=\int^\infty_{-\infty}\left\lbrace\psi'^\ast \psi-\frac{i}{\hbar}\delta t~\psi'^\ast \left(\hat{H}\psi\right)\right\rbrace_t d\mu=1
\eeq

We now consider the linear functional equation (Eq.~\eqr{func-amp1}) defined over the quotient measure spaces $X_{L,R}$. We will explicitly work with the left continuous quotient space $X_L$ and note that the results will analogously apply in $X_R$. We may then consider the linear functional to be defined by
\beq\label{func-amp2}
\braket{\psi'_{t+\delta t}|\psi_t}=\int^\infty_{-\infty}\left\lbrace\psi'^\ast \psi-\frac{i}{\hbar}\delta t~\psi'^\ast \left(\hat{H}\psi\right)\right\rbrace_t d\mu_L
\eeq

Inserting the Hamiltonian operator from~\eqr{schrodeq} into the transition amplitude~\eqr{amp1} and keeping linear terms in $\delta t$ yields
\begin{align}\label{amp2}
\braket{\psi'_{t+\delta t}| \psi_t}&=\braket{\psi'_t|\psi_t}-\frac{i}{\hbar}\delta t \bra{\psi'_t} -\frac{\hbar^2}{2m}\dfrac{d^2}{dx^2} +\alpha\delta '(x) \ket{\psi_t}+\ldots
\end{align}

For the moment, we work with the second term on the R.H.S. of~\eqr{amp2}. We wish to have in the space of test functions for this linear functional to cover all of $\mbb{R}$ or rather all the measure space $X_L$. We know the measure of $X_L$ is continuous with respect to the Lebesgue measure, $\lambda$. From~\ref{expl3} we see that we if we assume $\phi\in\mc{S}$, then we have the weak equivalences
\beq\label{sgn-delt-dist}
\braket{\delta'(x),\phi(x)}=-\braket{\delta(x),\phi'(x)}=-\frac{1}{2}\braket{\sgn'(x),\phi'(x)}=\frac{1}{2}\braket{\sgn(x),\phi''(x)}=-\phi'(0),
\eeq
up to negligible terms involving powers of $x$ multiplying $\delta(x)$\footnote{The terms of the form $x\delta(x)$ are discarded. These terms are either zero, or are orthogonal to the harmonic solution space $\mb{H}^0$.}. We furthermore make the assumptions that $\ket{\psi}$ is self-dual (i.e. $\bra{\psi}^\dagger=\ket{\psi}$). We use the above weak equivalence to make the identification of $\delta\sim_x\ast[\mathbold{d}\sgn_L]$ as in \secref{Hamcotan}. It is interesting to note that \textit{if} indefinite integrals of distributions were indeed defined, the same result could be obtained using integration by parts two times on second term on the R.H.S. of~\eqr{amp2}. 

\begin{rmk}\label{distalgebra}
It is well known that L-S measures do allow us to write the integral of $\delta_x$ similar to $\int_{\mbb{R}}f(x)\delta\left\{dx\right\}=\int_{-\infty}^{x}{} f(y)dH(y)$, for $f\in\mc{S}$ and $H(y)$, the Heaviside distribution. In fact Talvila \cite{Talvila2009,TalvilaLp} discusses Banach spaces of integrable distributions, where the above is defined uniquely. The issue which arises as that such spaces are not, in general, separable. However, this is not an issue for the present case. Throughout our derivation above and below, we assume that the spatial variable $x$ really represents some interval: $x\in(\cdot, \cdot]~or~x\in[\cdot,\cdot)$ of $\closure{\mbb{R}}^1$, and therefore admits a countable basis for the Hilbert space. The Banach spaces of integrable distributions are separable under such circumstances. However, this is no longer true once we take the \textit{"continuum limit"}. This is often the standard approach by physicists, but only after completing the \textit{operational calculi} steps. Thus, in theory, we could employ such methods as integration by parts and maintain the structure of a "functional" Hilbert space. It is likely that we would even have some notion of a non-commuting Banach algebra similar to the family of disentangling algebras $\left\{\mc{A}_t\right\}_{t>0}$ of \cite[Ch. 18]{johnlap}. However, we will leave such discussions for future work.
\end{rmk}

The $\ket{\psi}$ has an expansion in terms of some Schauder basis, such that $\ket{\psi}=\Sigma_n \ket{\psi_n}<\infty$. Therefore we can say that the unbounded differential operators and the Dirac-$\delta$ are weakly bounded for wave functions which belong to a compatible function space. The calculation is sketched as follows,  
\beq\label{amp2Htrm}
\bra{\psi'_{t}}\hat{H}\ket{\psi_t}&= \left[-\frac{\hbar^2}{2m}\frac{d^2}{dx^2} +\alpha\delta '(x)\right] \braket{\psi'_t |\psi_t},
\eeq

We have left continuity and apply the results of Prop.~\ref{Hequiv} to obtain an operator similar to~\cite{Branko1995}, which is just a linear transformation on the functional space\footnote{Theorem~\ref{muLR121BV} ensures that this map is well defined, as it is locally compact and Hausdorff on $(-\infty,\infty]$.}, 
\beq\label{Hpiece}
\bra{\psi'_{t}}\hat{H}\ket{\psi_t}&=\int^\infty_{-\infty}d\mu_L (x)\left[ -\frac{\hbar^2}{2m} +\frac{\alpha}{4}\sgn_L (x)\right]\frac{d^2}{dx^2}\braket{\psi'_t|\psi_t}.
\eeq

In order to reduce the accumulation of constants in~\eqr{Hpiece}, we relabel the constant terms with the definition $a\equiv\frac{\hbar}{\sqrt{2m}}$, and write the functional equation in the less cluttered form
\begin{align}\label{Hpiece2}
\bra{\psi'_{t}}\hat{H}\ket{\psi_t}&=\int^\infty_{-\infty}d\mu_L (x)\left[ -\left(a^2-\frac{\alpha}{4}\sgn_L(x)\right)\frac{d^2}{dx^2}\right]\braket{\psi'_t|\psi_t}\\
&=\int^\infty_{-\infty}d\mu_L (x)\left[a^2-\frac{\alpha}{4}\sgn_L(x)\right]\left(i\frac{d}{dx}\right)^2\braket{\psi'_t|\psi_t},
\end{align}
where we have adsorbed the minus sign by restoring the $i$ in the differential operators to make them Hermitian. 

We may now write the transition amplitude of~\eqr{amp2} to $1^{st}$ order in $\delta t$ as
\begin{align}\label{transamp}
\braket{\psi'_{t+\delta t} |\psi_t}&=\int^\infty_{-\infty}d\mu_L(x)\left[1-\frac{i}{\hbar}\delta t \hat{H}(x)\right]\braket{\psi'_t|\psi_t}+\mathcal{O}(\delta t)^2,\\
\shortintertext{where}\label{Ht}
\hat{H}(x)&=\left(a^2-\frac{\alpha}{4}\sgn_L(x)\right)\left(i\frac{d}{dx}\right)^2.
\end{align}
Eq.~\eqr{Ht} then represents the transformed connection on the fibers of the cotangent bundle. It is worth noting that the Fourier transform of Eq.~\eqr{transamp} above, (and more generally Eq.~\eqr{redH}) is very similar, and seemingly analogous to the difference equation results of \cite[Eq.~4.4]{TakFaddSpect2015} (with understandably different boundary conditions). Another notable similarity of the above result, is to the Hartee equation for infinitely many particles, where the $L^2$ well-posedness of which was discussed in terms of Strichartz estimates by \cite{LewinSabinHartee2015}.  

We close this section by noting that we may find the Lagrangian density function from Eq.~\eqr{Ht}, which defines an isomorphism of the fibers from the cotangent bundle to the tangent bundle. We recall the normalization condition of Eq.~\eqr{deltanorm} and interpret the $(x'-x)$ factor as a velocity by writing it as
\begin{equation}
\begin{aligned}\label{xdot}
x'-x&=\dfrac{dx}{dt}\delta t\\
&=\dot{x}\delta t
\end{aligned}\end{equation}
to first order in $\delta t$. By including a factor of $\hbar$, we may also interpret the differential operator as the momentum operator $\hat{P}$ such that
\beq\label{Phat}
\hat{P}=i\hbar\frac{d}{dx},
\eeq
and obeys the eigenvalue equation
\beq\label{Peigen}
\hat{P}\ket{\psi}=p\ket{\psi}
\eeq
with eigenvalue $p$. Equations \eqr{transamp} and~\eqr{Ht} are then a Legendre transformation on the linear functional equation. Together they yield, 
\begin{align}\label{T1}
\braket{\psi'_{t+\delta t} |\psi_t}&=\int^\infty_{-\infty}d\mu_L(x)\left\{\frac{i}{\hbar}\int^{t_f}_{t_i}\delta t~\mscr{L}\right\}+\mathcal{O}(\delta t)^2,\\
\shortintertext{where the Lagrangian density is}\label{expT}
\mscr{L}&=\dot{x}p-\left(-\frac{\hbar^2}{2m} + \frac{\alpha}{4}\sgn_L(x)\right)\cdot p^2.
\end{align}
Obviously, analogous results are obtained for the measure space $X_R$. 

Eq.~\eqr{expT} may be exponentiated, and inserted into the free Feynman functional integral,
\beq\label{FeynInt}
\braket{\psi_{t+\delta t}|\psi_t}=\int{\mc{D}x\mc{D}p}\exp^{\frac{i}{\hbar}\int{\delta t}\mscr{L}},
\eeq
where $\mc{D}x$ is the Feynman measure. The full path integral may be evaluated by first integrating over the momentum. Then analytically continue by $t\to i\tau$, which compactifies $\tau$ on $S^1$, and produces the convergent Gaussian integral.

\subsection{The domain (test function space) of the functional equation}\label{domwf}
 From Eqs.~\eqr{transamp} and~\eqr{Ht}, the topological measure space $X_L$, from which we demanded topological continuity in our solution space. Since a discontinuous function cannot be in $\mc{D}_\delta$, this leaves us out of the space of $L^p(\closure{\mbb{R}})$ functions. However as we have seen above, the semicontinuous quotient spaces allows the partial  embedding of measure extendable $L^p$ functions into $C_{L,R}(\closure{\mbb{R}})$. This is an artifact of the Lebesgue-Stieltjes measure when continuity is restricted to semicontinuity. There are a few more interesting properties of this which we will comment on shortly. However with respect to $X_{L,R}$, the Hamiltonian operator in Eq.~\eqr{Ht} is topologically continuous. Thus in the case of Lebesgue-Stieltjes measures, functions which are topologically continuous with respect to $X_{L,R}$, may now also be in $\mc{D}_\delta$. It follows that we may define a Sobolev space which has some $L^p$-space functions, but are continuous with respect to $X_{L,R}$. 
 
With Lebesgue-Stieltjes measures $\mu_{L,R}$, $\mc{D}_\delta$ includes the space of $C_{L,R}$ functions functions. In terms of the Hamiltonian operator of Eq.~\eqr{Ht}, $\mc{D}_{\hat{H}}=\{f|f''\in C_{L,R}(X_{L,R})\}:=C^2_{L,R}(X_{L,R})$, or the class of $\mu_{L,R}$-extendable $L^p$ functions which are also $\mu_{L,R}$-measurable (integrable) on all compact subsets of $X_{L,R}$. The dual space of $C^2_{L,R}(X_{L,R})$ is the spaces of left/right continuous measures of bounded variation, which is the completion of $C_{L,R}$ with respect to the $\|\cdot\|_{\mc{BV}}$-norm. We note that the Hamiltonian operator (Eq.~\eqr{Ht}) with the $\|\cdot\|_{op}$-norm is not only bounded (weakly) by the $\mc{BV}_{L,R}$-norm completion of $L^1_{L,R}(X_{L,R})\supset L^p_{L,R}(X_{L,R})$ for $1\leq p\leq\infty$, but in fact they are equivalent. Let $\phi\in C^2_{L,R}(X_{L,R})$ such that $\|\phi''\|_{L^1_{L,R}}=1$ and $\|V(\phi'')\|_{L^1_{L,R}}=\left|a\right|\|\phi\|_{L^1_{L,R}}$, for some $a\in\mbb{R}$ and $\left|a\right|<\infty$. Then
\beq\label{Hbound}
\|\hat{H}\phi\|_{L^1_{L,R}}&\leq\left|\sup\left|\phi''\right|+\sup\left|V(\phi'')\right|\right|\\
&\leq\sup\left|\|\phi''\|+|a|\|\phi''\|\right|\\
&\leq\left|\phi''\right|+|a|\left|\phi''\right|\\
&=\left(1+|a|\right)\left|\phi''\right|\\
&=\left(1+|a|\right)\cdot 1\\
&=\|\phi''\|_{\mc{BV}_{L,R}}.
\eeq
Then we have that the operator norm of the Hamiltonian is given by
\beq\label{Hop-norm}
\|\hat{H}\|_{op}&=\frac{\|\hat{H}\phi\|_{L^1_{L,R}}}{\|\phi''\|_{\mc{BV}_{L,R}}}\\
&=\frac{\left(1+|a|\right)\left|\phi''\right|}{\|\phi''\|_{\mc{BV}_{L,R}}}\\
&=\frac{\left(1+|a|\right)\cdot 1}{\left(1+|a|\right)}\\
&=1.
\eeq
 Moreover, if $\phi$ is Lipschitz continuous such that for a real number $M\geq 0$ with $\frac{\|\phi^{(n)}\|}{\|\phi\|}\leq M$ for all $n>0$, with $\|\phi\|=1$ and $|a|\leq M$, the previous norm-bounded result can be strengthened to $\|\hat{H}\phi\|_{L^1_{L,R}}\leq M\left(1+ |a|\right)=M\|\phi\|_{\mc{BV}_{L,R}}$. Inserting the Lipschitz result in the set of Eqs.~\eqr{Hop-norm}, the same result is obtained under a stronger form of continuity. 

 \subsection{The Hilbert space of $C_{L,R}\left(X_{L,R}\right)$}
 In the last section we saw that the $\|\cdot\|_{\mc{BV}_{L,R}}$-norm is the norm completion of $C_{L,R}\left(X_{L,R}\right)$ with $\|\cdot\|_{\sup}$-norm, and that $\mc{D}_{\hat{H}}=C^2_{L,R}\left(X_{L,R}\right)$. It follows that the wave function $\ket{\psi}$ must also belong to this Sobolev-type space, or belong to a completely continuous function space after two derivatives, which could be described as the measurable (integrable) functions of $C^2_L(X_L)\cap C^2_R(X_R)=C_{\lambda}(X_\lambda)$ on all compact subsets of $\closure{\mbb{R}}$. 
 
 In order to have a Hilbert space, any wave function $\ket{\psi}$ which satisfies these conditions should be self-dual, and the functional equation $\braket{\hat{H},\psi^\ast\psi}\sim\braket{\hat{H},\psi^2}$ should also satisfy the 
H\"{o}lder inequality with $\frac{1}{p}+\frac{1}{q}=1$. It follows that quantum mechanics requires we identify
\beq\label{Holder}
\braket{\hat{H},\psi^\ast\psi}^{\frac{1}{2}}=\braket{\hat{H},\psi^2}^{\frac{1}{2}}\sim\sqrt{\braket{\psi\left|\hat{H}\right|\psi}}.
\eeq   
With this in mind, it follows that we must restrict $\ket{\psi}$ to be those $C_{L,R}(X_{L,R})$ functions with Euclidean norm. Therefore we need a norm defined by $\|\cdot\|_{\mc{BV}_{L,R}}\cap\|\cdot\|_{L^2_{L,R}}$. This implies that the Hilbert space norm of $\ket{\psi}$ is given by $\sqrt{\|\cdot\|_{\mc{BV}_{L,R}}}$. Therefore we define the Hilbert space to be the following.

\begin{defin}{: The semicontinuous Hilbert space of $C^2_{L,R}\left(X_{L,R}\right)$}\label{Hilbertspace}\\
Let the Hamiltonian operator be given by Eq.~\eqr{Ht}, and $\psi\in C^2_{L,R}\left(X_{L,R}\right)$. The Hilbert space for $\psi$ is defined by the twice differentiable semicontinuous space of left/right Lebesgue-Stieltjes measurable (integrable) functions with $\mc{BV}_{L,R}$ 2-norm, or $\mc{H}:=\{\psi\big|\psi''\in C_{L,R}\left(X_{L,R}\right),~\text{and}~\sqrt{\norm{(\psi^\ast,\hat{H}\psi)}_{BV_{L,R}}}<\infty \}$. 
\end{defin}

Since we have a Hilbert space which is dependent upon the measure of an operator on $C^2_{L,R}\left(X_{L,R}\right)$ functions, this is becomes equivalent to the Schatten norm over these function spaces. In light of Ex.~\ref{expl3} and as a result of the semicontinuous quotient space construction, we have the following result, which by now may be obvious.
\begin{thm}{: Decomposition of $\mc{H}$}\label{Hdecomp}\\
The Hilbert space of $\mc{H}$ of $C^2_{L,R}\left(X_{L,R}\right)$-functions has the orthogonal decomposition of left (resp. right) measurable functions such that $\mc{H}=\mc{H}_L\oplus\mc{H}_R$, where $\mc{H}_L$ and $\mc{H}_R$ are separable orthogonal subspaces and submanifolds of left/right $\mu_L/\mu_R$-measurable (integrable) functions on the measure spaces $X_L$ and $X_R$, respectively. Moreover, for any non-atomic completely discontinuous function $f$ (neither left, nor right semicontinuous) over an interval $I=(\cdot,a)\cup(a,b)\cup(b,\cdot)\subset \mbb{R}$ which is Lebesgue measurable for $f$ over the subinterval $I_1=(a^+,b^-)$ such that $\lambda(f)_{(a^+,b^-)}\neq 0$, finite and with Lebesgue-Stieltjes measure $\mu_L(f)_I=\mu_R(f)_I=0$, is separately left and right $\mu$-extendable to the intervals $I_L=(a,b]\cup(b,\cdot]$ and $I_R=[\cdot,a)\cup[a,b)$, where $\mu_L(f)_{I_L}\neq 0$ implies $f_{I_L}\in\mc{H}_L$, and $\mu_R(f)_{I_L}=0$ and $\mu_R(f)_{I_L}\notin \mc{H}_L$. Similarly, $\mu_L(f)_{I_R}=0$ and $\mu_R(f)_{I_R}\neq 0$ implies $\mu_R(f)_{I_R}\in\mc{H}_R$ and $\mu_R(f)_{I_L}\notin\mc{H}_R$. Therefore each left/right extension is separately unique in $\mc{H}_L$ and $\mc{H}_R$.
\end{thm}
\begin{proof}{:} The proof of this can most easily be seen by first referring to Ex.~\ref{expl3}. This shows for a simple step function (Heaviside function) that there is a unique left and right extension of the Heaviside function such that $\mu_L(H_L)\perp\mu_R(H_R)$. Separability of $\mc{H}_L$ and $\mc{H_R}$ is inherited by the countable measure topology basis of $X_{L,R}$. The left/right extension is chosen such that the function is piecewise extended to the left/right by a set of Lebesgue measure zero, and the left/right function value over the interval extension is continuous (i.e. with no jump). Given any $\mu_{L,R}$ extendable function $f$ in the spaces $C_{L,R}(X_{L,R})$ (which by construction may not consist only of atoms), $f$ may be approximated by a sequence of step functions semicontinuous step functions $\chi_{L_n}\in C_L(X_L)$ or $\chi_{R_n}\in C_R(X_R)$. There are two possible cases at each boundary point, and a third separate case which we will explain after the boundary point cases. 

Case i) Take a sequence of left continuous step functions. Each left sequence is, be definition zero over any boundary point which is discontinuous from the left \textbf{and} lies outside the interval extension (i.e. if the extended interval $\closure{\chi_{(a,b)}}=\chi_{(a,b]}$, this by definition implies $\chi(a)=0$ unless there is a different step function defined over the interval $\chi_{(\cdot,a]}$). The analogous holds for a right continuous sequence of step functions.

Case ii) Take a completely discontinuous function $f$ over the joined intervals $I=(\cdot,a)\cup(a,\cdot)$, where a jump occurs at $f(a)$. Choose to extend $f$ such that it becomes left continuous, $\widehat{f}=f_L$ over the interval $I_L=(\cdot,a]\cup(a,b]$, such that $f(a)\neq f(a^+)$. In the topological measure space $X_L$, $\mu_L(f(x=a))$ is well defined, however $\mu_R(f=a)=\emptyset$, and $\mu_R(\emptyset)=0$. The analogous holds for $f$ right extended, such that $\widehat{f}=f_R$ on $I_R=[\cdot,a)\cup[a,\cdot)$. Since in either case, there is a jump discontinuity at $f(a)$, and $\widehat{f}$ is chosen to be continuous from either the left or the right, we have that $f_L(a)\neq f_R(a)$, otherwise there would be no jump, and thus each left/right extension is unique.
 
 In either Case i) or ii), we have that either $\mu_L(f)=\emptyset$ and $\mu_R(f)=\emptyset$, which implies $\mu_L(f)=\mu_R(f)=0$, or $\widehat{f}=f_L$, which implies $\mu_L(f)\neq 0$ and $\mu_R(f)=0$ (and similarly for $\widehat{f}=f_R$). Again, the left/right extension is unique.
 
 Case iii) There is a function $g$ over some interval $I_0$, for which $\mu_L(g)_{I_0}=\mu_R(g)_{I_0}=0$ and $g$ is not $\mu_{L,R}$-extendable. In this case $g\sim[0]$ (the equivalence class of the zero function for both $X_L$ and $X_R$). However $[0]$ is the only function equivalence class which may be common to both $X_L$ and $X_R$.

The remaining aspects of the proof follow straightforwardly.
\end{proof} 

\section{Indefiniteness of $\mc{H}$ and Krein Spaces}\label{Krein}
 In this section we only wish to make some cursory comments regarding the formalism developed above and the theory of Krein spaces (and Krein space operators). For a concise overview of Krein spaces see~\cite{Rovnyak2002}. For a more comprehensive introduction see~\cite{ellis2003}. 
 
 \subsection{Krein spaces and Krein space operators}\label{Kreinspaces}
 We summarize some basic definitions of Krein spaces and Krein space operators given by~\cite[Sec. 3]{Rovnyak2002}. Let $\mf{H},\mf{K}$ denote Krein spaces on $\closure{\mbb{R}}$. A Krein space (which may also be a Pontryagin space) is an indefinite inner produce space which is representable as the orthogonal direct sum $\mf{H}=\mf{H}_+\oplus\mf{H}_-$, where $\mf{H}_+=\set{\left(\mf{H},\braket{\cdot |\cdot}\right)}$ is a positive-definite Hilbert space and $\mf{H}_-=\set{\left(\mf{H},-\braket{\cdot |\cdot}\right)}$ is the \textit{antispace} of a Hilbert space with a negative inner product. A fundamental symmetry on $\mf{H}$ are symmetries expressible an orthogonal direct sum. The Hilbert space topology is the strong topology induced on $\mf{H}$, and the $\dim\mf{H}_\pm$ are the indices of $\mf{H}$. A Pontryagin space is a Krein space with finite $\mf{H}_-$ index. 
 
 The spaces of continuous linear functionals (operators) and adjoint operators are denoted by $\mc{L}\left(\mf{H}\right)$ and $\mc{L}\left(\mf{f},\mf{K}\right)$. For some operator $A\in\mc{L}\left(\mf{H},\mf{K}\right)$ then $A^\ast\in\mc{L}\left(\mf{K},\mf{H}\right)$, with $\braket{Af,g}=\braket{f,A^\ast g}$ for some $f\in\mf{H}$ and $g\in\mf{K}$. 
 \begin{defin}{: Properties of Krein space functions}\label{Kfuncs}\\
 Let $A\in\mc{L}(\mf{H})$. Then $A$ is:
 \item\hspace{.5cm} \textit{self adjoint} if $A^\ast=A$,
 \item\hspace{.5cm} \textit{a projection} if $A^\ast=A$ and $A^2=A$,
 \item\hspace{.5cm} \textit{nonnegative} if $\braket{Af,f}\geq 0,~\forall~f\in\mf{H}$.
 \end{defin} 
 
 \begin{defin}{: Properties of Krein space operators}\label{Kops}\\
 Let $A\in\mf{H}$ be self adjoint, and denote the supremum of all $r$ for which there exists an $r$-dimensional subspace of $\mf{H}$ that is a (anti)Hilbert space by $\text{ind}_+A$, respectively $\text{ind}_-A$, in the inner product given by $\braket{f,g}_A=\braket{Af,g}$ for $f,g\in\mf{H}$. Let $B$ be an operator in $\mc{L}\left(\mf{H},\mf{K}\right)$, then $B$ is:
 \item\hspace{.5cm} \textit{isometric} if $B^\ast B=1_\mf{H}$,
 \item\hspace{.5cm} \textit{partially isometric} if $BB^\ast B=B$,
 \item\hspace{.5cm} \textit{unitary} if both $B$ and $B^\ast$ are isometric,
 \item\hspace{.5cm} \textit{a contraction} if $B^\ast B\leq 1_{\mf{H}}$,
 \item\hspace{.5cm} \textit{a bicontraction} if both $B$ and $B^\ast$ are contractions.
 \end{defin}
 
 We also note a difference between Krein spaces ($\mf{H}$) and Hilbert spaces ($\mc{H}$) regarding orthogonality. Let $\mf{M}\subset\mf{H}$ be a closed subspace in $\mf{H}$. It is generally not true that $\mf{H}=\mf{M}\oplus\mf{M}^\perp$. However, if in addition to being a linear subspace of $\mf{H}$, $\mf{M}$ is also a regular subspace (a Krein subspace) if it is closed and a Krein space in the inner product of $\mf{H}$. If these conditions hold, then we have the following: 
 \begin{defin}{: Krein subspaces}\label{subKrein}\\
 Let $\mf{M}$ be a regular subspace of $\mf{H}$, then the following are equivalent:
 \item\hspace{.5cm} $\mf{M}$ is a Krein subspace,
 \item\hspace{.5cm} $\mf{H}=\mf{M}\oplus\mf{M}^\perp$,
 \item\hspace{.5cm} For a projection operator $\hat{\mc{P}}\in\mc{L}\left(\mf{H}\right)$ such that $\hat{\mc{P}}:\mf{H}\to\mf{M}$, then $\mf{M}=\text{ran}~\hat{\mc{P}}$.
 \end{defin} 
 
 \subsection{$\mc{H}$ as a Krein space}\label{H-Kspace}

From what we have seen in~\secref{Kreinspaces}, the Hilbert space certainly has the properties of a Krein space. The Hamiltonian functional operator Eq.~\eqr{Ht} on the orthogonal measure spaces $X_{L,R}$ provides a decomposition of $\mc{H}$. Let $\psi_L,\phi_L\in\mc{H}_L$ and $\psi_R,\phi_R\in\mc{H}_R$. Since Eq.~\eqr{Ht} was found explicitly on the measure space $X_L$, we denote the left/right Hamiltonian by $\hat{H}_L=$ Eq.~\eqr{Ht}, and $H_R=$ Eq.~\eqr{Ht} with $\sgn_L(x)\to\sgn_R(x)$, and denote the initial Hamiltonian function $\hat{H}=$ Eq~\eqr{amp2Htrm}. Then by repeating the steps from moving from Eq.~\eqr{amp2Htrm} to Eq.~\eqr{Ht} in the case of $X_R$ and denoting the eigenvalues of $\hat{H}_{L,R}\ket{\psi_{L,R}}=E_{L,R}\ket{\psi_{L,R}}$, we have the functional result
\beq\label{HpsiLR}
\braket{\psi|\hat{H}|\psi}&=\braket{\psi_L|\hat{H}|\psi_L}\oplus\braket{\psi_R|\hat{H}|\psi_R}\oplus\braket{\psi_L|\hat{H}|\psi_R}\oplus\braket{\psi_R|\hat{H}|\psi_L}\\
&=\braket{\psi_L|\hat{H}_L\psi_L}\oplus\braket{\psi_R|\hat{H}_R\psi_R}\oplus\braket{\psi_L|\hat{H}_R\psi_R}\oplus\braket{\psi_R|\hat{H}_L\psi_L}\\
&=E_L\braket{\psi_L|\psi_L}\oplus E_R\braket{\psi_R|\psi_R}\oplus 0\cdot\braket{\psi_L|\psi_R}\oplus 0\cdot\braket{\psi_R|\hat{\psi}_L}\\
&=E_L\braket{\psi_L|\psi_L}\oplus E_R\braket{\psi_R|\psi_R}
\eeq
where orthogonality of $\ket{\psi_L},\ket{\psi_R}$ is used from the second to the the third lines. 

Eq.~\eqr{HpsiLR} shows that we have the Hilbert space $\mc{H}=\left(\mc{H}_L,E_L\braket{\psi_L|\psi_L}\right)\oplus\left(\mc{H}_R,E_R\braket{\psi_R|\psi_R}\right)$. However $\mc{H}$ being expressible as an orthogonal decomposition does not ensure that $\mc{H}$ is itself is Krein.
A necessary condition is that $\mc{H}_{L,R}$ are each a Krein subspace. Therefore we must be able to show that $\mc{H}_{L,R}=\mc{H}_{L,R,+}\oplus\mc{H}_{L,R,-}$, where $\mc{H}_{L,R,+}$ is a Hilbert space and $\mc{H}_{L,R,-}$ is the associated anti-Hilbert space. We already have the indefinite structure built into our left/right states via the coupling constant $\alpha$. The Hamiltonian functional was defined such that the sign of $\alpha$ was unspecified.

 Let $\mf{H}$ be a Krein space associated with the Hilbert space $\mc{H}$, $A$ a definitizable operator in $\mf{H}$, and $E$ the spectral function of $A$. If $A$ is positive, its spectrum is $\sigma(A)\in\mbb{R}$, where 0 is the only non-negative semi-simple eigenvalue of $\sigma(A)$. Also $\set{0,\infty}$ may be the only critical points of the spectrum. If $\set{0,\infty}$ are regular critical points (non-singular), and 0 is not an eigenvalue, then $A^\ast=A$ in the Hilbert space $\mc{H}=\left(\mf{H}, \braket{(E(\mbb{R}_+)-E(\mbb{R}_-))\cdot,\cdot}\right)$. 
 
 It was shown in~\cite{Branko1995} the if $\mf{H}=L^2(\mbb{R})$ with $y\in\mc{D}_A=W^{2,2}$, where $W^{2,2}$ is the $L^2$ Sobolev space, then the Julian operator $Jy=\sgn(x)\frac{d^2}{dx^2}y$ is a fundamental symmetry on $\mf{H}$. Then $Ay=-\sgn(x)\frac{d^2}{dx^2}y$ is congruent to a self adjoint operator on $L^2(\mbb{R})$. Moreover $A$ has no eigenvalues, $\sigma(A)=\mbb{R}$, and $\set{0,\infty}$ are regular and the only critical critical points of $\sigma(A)$. Given the work of~\cite{Branko1995}, we can directly infer the following about the Hamiltonian functional operator Eq.~\eqr{Ht} and the Hilbert (anti)spaces. 
\begin{thm}{: Properties of $\hat{H}$ as a Krein space $\mf{H}$}\label{Hop-props}\\
Let $\hat{H}$ be the Hamiltonian functional operator given by Eq.~\eqr{Ht} on the Hilbert space $\mc{H}=\mc{H}_L\oplus\mc{H}_R$, and $\mf{H}=C_{L,R}\left(X_{L,R}\right)\supset L^2_{L,R}\left(X_{L,R}\right)$ with an indefinite inner product. Then $\mf{H}$ is direct sum decomposable as $\mf{H}=\left(\mc{H}_{L}\oplus\mc{H}_{R}\right)_+\oplus\left(\mc{H}_{L}\oplus\mc{H}_R\right)_-$, with 
\item $\mc{H}_{L,+}=\left(\mf{H}, \braket{E_L(\closure{\mbb{R}}_{(\cdot,\cdot]})\cdot,\cdot}\right)$ being the left continuous Hilbert space,
\item $\mc{H}_{R,+}=\left(\mf{H}, \braket{E_R(\closure{\mbb{R}}_{[\cdot,\cdot)})\cdot,\cdot}\right)$ being the right continuous Hilbert space,
\item $\mc{H}_{L,-}=\left(\mf{H}, -\braket{E_L(\closure{\mbb{R}}_{(\cdot,\cdot]})\cdot,\cdot}\right)$ being the left continuous Hilbert antispace,
\item and $\mc{H}_{R,-}=\left(\mf{H}, -\braket{E_R(\closure{\mbb{R}}_{[\cdot,\cdot)})\cdot,\cdot}\right)$ being the right continuous Hilbert antispace.\\
Moreover $\hat{H}$ is self adjoint on $\mc{H}_{L,R,+}$ and anti-self adjoint on $\mc{H}_{L,R,-}$.
\end{thm}
\begin{proof}{:}
Each Hilbert space and antispace $\mc{H}_{L,R,\pm}$ are mutually orthogonal and constructed from the quotient spaces of half-open topologies on $\closure{\mbb{R}}$. The Hilbert space topology is the strong topology inherited by $\mf{H}$, therefore $\mf{H}$ inherits the left/right continuity of $X_{L,R}$. Analogous arguments to~\cite{Branko1995} show that $\hat{H}$ is self adjoint on $\mc{H}_{L,R,+}$ separately when $\alpha>0$, which corresponds to the operator $A=-\sgn(x)\frac{d^2}{dx^2}$, defined in the previous paragraph. Since the Julian operator defined above by $J=\sgn(x)\frac{d^2}{dx^2}$ is a fundamental symmetry of $\mf{H}$ (again shown in~\cite{Branko1995}), which corresponds to the case of $\alpha<0$ and $-A=J$, it follows that $\mf{H}$ is anti-self adjoint on $\mc{H}_{L,R,-}$. 
 \end{proof}
 
 \section{Summary and Concluding Remarks}\label{conclusion}
 We have investigated the quantum mechanics Hamiltonian with a Dirac-$\delta'$ potential as a continuous linear functional operator. The particular aspect of this equation is that the kinetic energy operator is a diffeomorphism mapping from the space of weakly continuous linear functions $\mc{L}$ to another function space $\mc{L}''$ by $(D\circ D):\mc{L}\to\mc{L}'\to\mc{L}''$. However $\delta'$ may be considered (after one integration by parts) as a measure or distribution on the linear transformation from $\mc{L}\to\mc{L}'$, or rather $T:\mc{L}\to\mc{L}'\to\mbb{C}$ (or $\mbb{R}$). Moreover the space of test functions of $\delta$ and $\delta'$ is equivalent to any $L^p$ space due to $L^p$ containing discontinuous functions. 
 
In order to resolve the domain incompatibilities for distributional potentials in quantum mechanics, we constructed the spaces of semicontinuous functions. The spaces $L^P_{L,R}$ are projective subspaces of the standard $L^p$ spaces, which single point extensions/restrictions on each disconnected open set for which they are defined on $L^p$. This effectively allows all $L^p$ (aside from functions of atomic sets) functions to be identified as left/right semicontinuous  function, which is topologically continuous when defined on their corresponding topologically semicontinuous measure spaces, $X_{\mu_L}$ or $X_{\mu_R}$. This continuity is reduced to semicontinuity on each measure space. The $\|\cdot\|_{sup}$ bounds all $\|\cdot\|_{L^p_{L,R}}$ including the finitely additive measures of bounded variation, $\|\cdot\|_{BV_{L,R}}$. Under these conditions, we have that $L^p_{L,R}\hookrightarrow C_{L,R}$ is a partial embedding of semicontinuous $L^p$ functions into semicontinuous spaces $C_{L,R}$, with Riemann-Stieltjes integral. We may view Riemann-Stieltjes measures as an extension of the Riemann measure to include half-open intervals, or the Lebesgue-Stieltjes integral as a restriction of the Lebesgue measure to half-open intervals. In this way, they are equivalent on $C_{L,R}$ spaces. The $C_{L,R}$ spaces provide two advantages over the standard $C(\mbb{R})$ and $L^p$ spaces. The first advantage is that $C_{L,R}$ allows us define a common space of test functions for $\delta$ and $\delta'$ functionals which includes as subspaces semicontinuous $L^p$ functions. The second advantage is that we may include Banach spaces of regulated distributions, which invert regulated distributions in terms of their primitive functions.

In~\secref{deltaprm} we analyzed the functional Hamiltonian Eq.~\eqr{hamdelta2} on the semicontinuous spaces as differentiable manifolds, complete with the tangent (and cotangent) fiber bundle structures. We then obtained a connection form transformation of Eq.~\eqr{hamdelta2}, which was shown to be canonical on the cotangent bundle. This permitted an equivalence class identification, which was a foliation of the of the cotangent bundle in terms of the cohomology classes of linear functionals with derivative of their primitive functions. The fact that inverses of regulated distributions is possible in the semicontinuous function spaces is needed to make these equivalence identifications well defined. In that way, semicontinuity was the key property which made such constructions possible. The semicontinuity of the differentiable manifolds then allowed us to define a common domain of Eq.~\eqr{hamdelta2}, and determine a 0-form wave function solution exists in such a way it is in the cohomology class of harmonic 0-forms for both the kinetic energy operator and the $\delta'$ potential. The orthogonality of $L^p_{L,R}$ function spaces extends to the Hilbert spaces of the Hamiltonian operator Eq.~\eqr{Ht}, and thus provided an semicontinuous orthogonal decomposition of the Hilbert space $\mc{H}$.  
 
 In~\secref{Krein}, we discussed the Hilbert space $\mc{H}$ within the indefinite structure of Krein spaces, $\mf{H}$. The indefinite structure was implicitly manifest through indefinite multiplicative coupling $\alpha\sgn(x)$, defined in Eq.~\eqr{Ht}. Therefore we found that the Hilbert space and associated antispace were regular subspaces of $\mf{H}$. The work of~\cite{Branko1995}, shows that the Hamiltonian functional equation, Eq.~\eqr{Ht}, is self adjoint on $C^2_{L,R}$ for the case of $\alpha>0$ and anti-self adjoint for the case when $\alpha<0$. 
 
There remains open questions to which we leave for future work. In particular, the existence of an antispace of $\mc{H}$ implies the existence anti-particle states inherent in QFT. Here they are manifest in a basic quantum mechanics construction. A complete spectral analysis of the system provide insight between quantum mechanics and quantum field theory with respect to this system. What is interesting is that our construction was based on a classical Banach space formulation of quantum mechanics, yet it seems that notions of quantum field theory are almost implicit. Obviously restricting to the half-line would remove the negative definite components of the spectrum. However, that notion seems unsatisfying. There is nothing special regarding $\closure{\mbb{R}}^-$. Regardless, the coupling term $\alpha\sgn(x)$ is almost better viewed as a Pfaffian-like $\beta$-function in the QFT path integral quantization. This touches upon work current in progress regarding a Feynman path integral formulation of the system for which we will investigate (among other things) anomalous bound states, and ghost states resulting from the Feynman "measure", and the possibility of supersymmetric states. In regards to the latter, a supersymmetric field is defined via fields which obey an anti-symmetric commutator algebra. An investigation of some algebraic structure, as in Rmrk. \ref{distalgebra}, would be needed for such an analysis. 

In light of the discussion in~\secref{Krein} and the odd character of the Dirac-$\delta^\prime$ potential, supersymmetric states may be implicit in a very natural way. Preliminary calculations suggest this to be the case. Anomalous bound (or even possibly scattering) states may also provide some theoretical predictions regarding low dimensional solid state systems. For example, the possible quantization of magnetically induced current flows in carbon nano-tubes and other small scale structures for which quantum interactions become dominant. A Feynman path integral formulation of the Dirac-$\delta'$ system shows that the Feynman measure can introduce non-trivial dynamics through the exponential (i.e. ghosts), which can become dominant if the coupling constant is on the order of unity. 

Another intriguing component of our study here is the connection form Eq.~\eqr{redH} derived from Eq.~\eqr{hamdelta2}. Eq.~\eqr{redH} has a form similar in nature to the Dirac-Born-Infeld (DBI) operator. It would be interesting to generalize what has done hear to $\closure{\mbb{R}}^n$ and look  to see if this analogy indeed holds true. Intuitively, one would expect that 1-dimensional Pfaffians would be come components of spinors in higher dimensions. Also, the possibility of defining a consistent "Lie algebra" using this canonical form of Eq.~\eqr{redH}, and the relevant implications for an algebra representation for jets, or for pseudo-differential operators. It would be interesting to investigate the limits of our construction here in terms of these formalisms, both separately and in conjunction with the possible DBI operator connection. 

\subsection*{Acknowledgments}
The author would like to thank Prof. Helge Holden for helpful comments regarding historical developments and useful references for this work. The author is particularly grateful to Dr. Michael Maroun for his friendship, the many helpful comments, uncountable enlightening discussions, and suggested references. This work would not have been possible without his input. Finally, the author is also tremendously grateful to Dr. Tuna Yildirim for his friendship and willingness to help proof read this document for grammatical errors.

\bibliographystyle{amsplain}

\end{document}